\documentclass[sigconf, nonacm]{acmart}

\newcommand\vldbdoi{XX.XX/XXX.XX}
\newcommand\vldbpages{XXX-XXX}
\newcommand\vldbvolume{19}
\newcommand\vldbissue{1}
\newcommand\vldbyear{2026}
\newcommand\vldbauthors{\authors}
\newcommand\vldbtitle{\shorttitle} 
\newcommand\vldbavailabilityurl{https://github.com/L1ANLab/CD-SBN}
\newcommand\vldbpagestyle{plain} 

\usepackage{amsmath,amsfonts}
\usepackage[ruled,linesnumbered,vlined]{algorithm2e}

\usepackage{etoolbox}
\makeatletter
\patchcmd\algocf@Vline{\vrule}{\vrule \kern-0.4pt}{}{}
\patchcmd\algocf@Vsline{\vrule}{\vrule \kern-0.4pt}{}{}
\makeatother

\usepackage{graphicx}
\usepackage{textcomp}
\usepackage{xcolor}
\usepackage{enumerate}
\usepackage{subfigure}
\usepackage{booktabs}
\usepackage{balance}
\usepackage[normalem]{ulem}

\usepackage{tcolorbox}
\usepackage{marginnote}
\usepackage{hyperref}
\hypersetup{
colorlinks=true,
linkcolor=black,
anchorcolor=blue,
citecolor=black,
urlcolor=black
}

\newcommand{\nop}[1]{}

\newtheorem{definition}{Definition}
\newtheorem{example}{Example}
\newtheorem{lemma}{Lemma}

\begin{document}
\title{Efficient Community Detection Over Streaming Bipartite Networks}

\author{Nan Zhang}
\affiliation{%
  \institution{East China Normal University}
  \city{Shanghai}
  \country{China}
}
\email{52275902026@stu.ecnu.edu.cn}

\author{Yutong Ye}
\affiliation{%
  \institution{East China Normal University}
  \city{Shanghai}
  \country{China}
}
\email{52205902007@stu.ecnu.edu.cn}

\author{Xiang Lian}
\affiliation{%
  \institution{Kent State University}
  \city{Kent}
  \state{Ohio}
  \country{USA}
}
\email{xlian@kent.edu}

\author{Qi Wen}
\affiliation{%
  \institution{East China Normal University}
  \city{Shanghai}
  \country{China}
}
\email{51265902057@stu.ecnu.edu.cn}

\author{Mingsong Chen}
\affiliation{%
  \institution{East China Normal University}
  \city{Shanghai}
  \country{China}
}
\email{mschen@sei.ecnu.edu.cn}

\begin{abstract}
The streaming bipartite graph is widely used to model the dynamic relationship between two types of entities in various real-world applications, including movie recommendations, location-based services, and online shopping. Since it contains abundant information, discovering the dense subgraph with high structural cohesiveness (i.e., \textit{community detection}) in the bipartite streaming graph is becoming a valuable problem. Inspired by this, in this paper, we study the structure of the community on the butterfly motif in the bipartite graph. We propose a novel problem, named \textit{\underline{C}ommunity \underline{D}etection over \underline{S}treaming \underline{B}ipartite \underline{N}etwork} (CD-SBN), which aims to retrieve qualified communities with user-specific query keywords and high structural cohesiveness at \textit{snapshot} and \textit{continuous} scenarios. In particular, we formulate the user relationship score in the weighted bipartite network via the butterfly pattern and define a novel $(k,r,\sigma)$-bitruss as the community structure. To efficiently tackle the CD-SBN problem, we design effective pruning strategies to rule out false alarms of $(k,r,\sigma)$-bitruss and propose a hierarchical synopsis to facilitate the CD-SBN processing. We develop efficient algorithms to answer snapshot and continuous CD-SBN queries by traversing the synopsis and applying pruning strategies. With extensive experiments, we demonstrate the performance of our CD-SBN approach on real/synthetic streaming bipartite networks. 
\end{abstract}

\maketitle

\pagestyle{\vldbpagestyle}
\begingroup\small\noindent\raggedright\textbf{PVLDB Reference Format:}\\
\vldbauthors. \vldbtitle. PVLDB, \vldbvolume(\vldbissue): \vldbpages, \vldbyear.\\
\href{https://doi.org/\vldbdoi}{doi:\vldbdoi}
\endgroup
\begingroup
\renewcommand\thefootnote{}\footnote{\noindent
This work is licensed under the Creative Commons BY-NC-ND 4.0 International License. Visit \url{https://creativecommons.org/licenses/by-nc-nd/4.0/} to view a copy of this license. For any use beyond those covered by this license, obtain permission by emailing \href{mailto:info@vldb.org}{info@vldb.org}. Copyright is held by the owner/author(s). Publication rights licensed to the VLDB Endowment. \\
\raggedright Proceedings of the VLDB Endowment, Vol. \vldbvolume, No. \vldbissue\ %
ISSN 2150-8097. \\
\href{https://doi.org/\vldbdoi}{doi:\vldbdoi} \\
}\addtocounter{footnote}{-1}\endgroup

\ifdefempty{\vldbavailabilityurl}{}{
\begingroup\small\noindent\raggedright\textbf{PVLDB Artifact Availability:}\\
The source code, data, and/or other artifacts have been made available at \url{\vldbavailabilityurl}.
\endgroup
}

\begin{figure}[t]
   \vspace{-2ex}
   \subfigcapskip=-1ex
    \centering
    \subfigure[$G_{t-1}$ at time $(t-1)$]{
        \includegraphics[height=2.5cm]{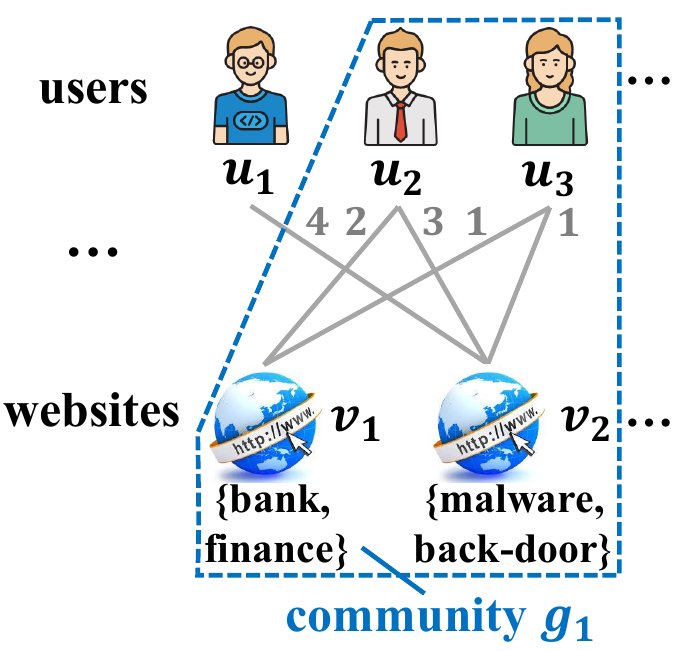}
        \label{subfig:motivation_1}
    }
    \subfigure[$G_t$ at time $t$]{
        \includegraphics[height=2.5cm]{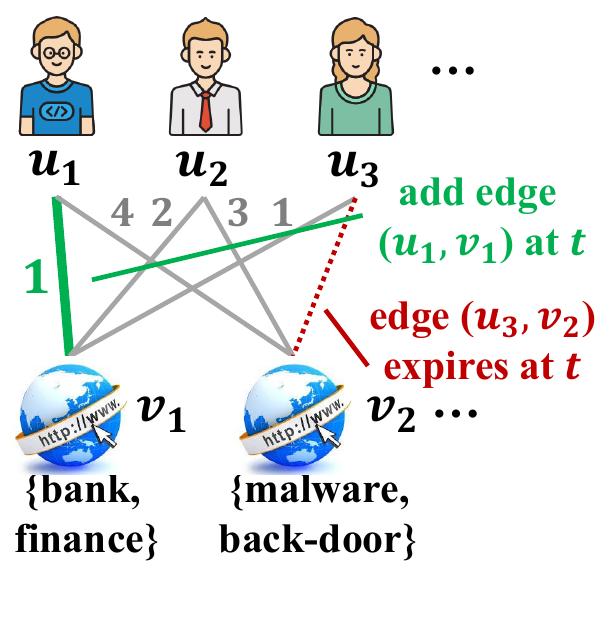}
        \label{subfig:motivation_2}
    }
        \subfigure[$G_{t+1}$ at time $(t+1)$]{
        \includegraphics[height=2.5cm]{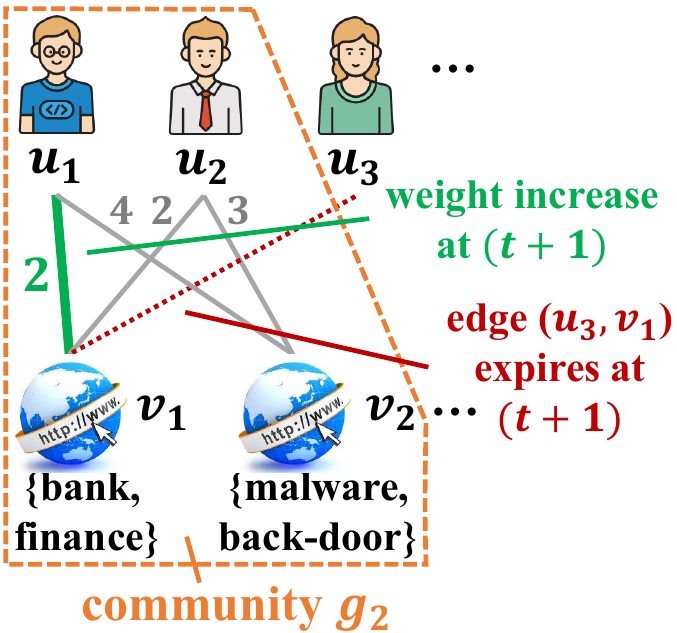}
        \label{subfig:motivation_3}
    }
    \vspace{-4ex}
    \caption{\small A community detection example in streaming click bipartite networks $G_t$.}
    \vspace{-4ex}
    \Description{An example of community detection in streaming click bipartite networks $G_t$.}
    \label{fig:motivation}
\end{figure}

\section{Introduction}
Nowadays, the management of bipartite graphs that involve two distinct types of nodes (e.g., customers and products, visitors and check-in locations, users and clicked websites) has been extensively studied in many real applications such as online recommendation from user-product transaction graphs \cite{fan2022ModelingUserBehavior,weng2022distributed,kou2023ModelingSequentialCollaborative}, location-based services with visitor-location check-in graphs \cite{hang2018ExploringStudentCheckIn,zhou2021ButterflyCountingUncertain,wang2017SemanticAnnotationPlaces,kim2022abc}, and behavior analysis in user-webpage click graphs \cite{sun2009RankingbasedClusteringHeterogeneous,neria2017risk}.



The \textit{community detection} (CD) over such bipartite graphs has recently become increasingly important for bipartite graph analysis and mining \cite{barber2007modularity,bouguessa2020BiNeTClusBipartiteNetwork,yen2020community,zhou2019analysis}, which uncovers hidden and structurally cohesive bipartite subgraphs (or groups). The CD problem over bipartite graphs has numerous useful real-world applications, including online product recommendations and marketing, user behavior analysis, and identifying malicious user groups.

Although many prior works \cite{wang2023DiscoveringSignificantCommunities,abidi2023SearchingPersonalizedkWing,wang2021EfficientEffectiveCommunity,zhang2024SizeboundedCommunitySearch,wang2022EfficientPersonalizedMaximum,zhang2021ParetooptimalCommunitySearch} studied the CD problem in \textit{static} bipartite graphs, real-world bipartite graphs are usually dynamically changing over time (e.g., updates of purchase transactions, check-ins, or website browsing). Therefore, in this paper, we will consider the \textit{\underline{C}ommunity \underline{D}etection over \underline{S}treaming \underline{B}ipartite \underline{N}etwork} (CD-SBN), upon graph updates (e.g., edge insertions/deletions or edge weight changes). Below, we provide a motivating example of our CD-SBN problem to detect malicious user groups (communities) on a user-website click bipartite graph.

\begin{example} 
{\bf{(Anomaly Detection for the Cybersecurity)}} In the application of online anomaly detection, the network security officer may want to specify query keywords for websites (e.g., ``bank'' and ``malware'') and identify a group (community) of malicious users accessing sensitive websites. Users in communities frequently visit common websites with query keywords (e.g., bank or hacking websites). The officer will warn suspicious users in the resulting community and take immediate actions (e.g., record evidence for the police).

Figure~\ref{subfig:motivation_1} illustrates an example of a clickstream bipartite network $G_{t-1}$ at timestamp $(t-1)$, which contains 3 user vertices, $u_1\sim u_3$, 2 website vertices, $v_1$ and $v_2$, and 5 edges $(u_i, v_a)$ between user $u_i$ and website $v_a$ (for $1\leq i\leq 3$ and $1\leq a \leq 2$). Here, each website vertex $v_a$ has a set of keywords that represent the website's features (e.g., \{bank, finance\} for a bank website $v_1$). Moreover, each edge $(u_i, v_a)$ is associated with an integer weight, indicating the frequency that user $u_i$ accessed website $v_a$ over the past five minutes. At timestamp $(t-1)$, the officer notices a community $g_1$ (circled by blue dashed lines) and starts to monitor suspicious users $u_2$ and $u_3$ in $g_1$.

At timestamp $t$, a new access by user $u_1$ to website $v_1$ (i.e., edge $(u_1, v_1)$) is recorded. Meanwhile, since the officer only focuses on suspicious users who frequently visit sensitive websites within five minutes (rather than normal users who visit infrequently over a long period), access records five minutes ago will expire (i.e., edge $(u_3, v_2)$ is deleted). Due to this edge deletion, community $g_1$ is no longer cohesive (i.e., less suspicious), so the officer will pay less attention to $g_1$.


Finally, at timestamp $(t+1)$, due to frequent accesses by user $u_1$ to website $v_1$ and the expiration of the access by user $u_3$ to $v_1$, the officer will pinpoint a new community $g_2$ (as circled by orange dashed lines in Figure~\ref{subfig:motivation_3}) and take actions against users and websites in $g_2$. \quad $\blacksquare$

\end{example}

In addition to the example above, the CD-SBN problem has many other real applications. For example, in the customer-product transaction bipartite graph, we can utilize the CD-SBN results to identify communities of customers who have recently exhibited purchasing behaviors related to online advertising and marketing. Similarly, in the visitor-location check-in bipartite graph (e.g., from Yelp \cite{yelp2024}), we can identify a group of visitors who frequently check in at some common \textit{points of interest} (POIs) and provide them with group buying coupons/discounts (e.g., Groupon \cite{groupon2024}).

In this paper, we formally define the community semantics in the context of dynamic bipartite graphs, that is, a keyword-aware and structurally dense bipartite subgraph (called $(k,r,\sigma)$-bitruss containing query keywords, as described in Section~\ref{subsec:community_detection_over_streaming_bipartite_network}). Then, we will tackle the CD-SBN problem, which identifies all bipartite communities with user-specified vertex features and high structural cohesiveness. We consider both \textit{snapshot} and \textit{continuous} scenarios for our CD-SBN problem. The former detects communities over a snapshot of the streaming bipartite network $G_t$ at timestamp $t$, whereas the latter continuously monitors CD-SBN answers for each registered community constraint upon streaming graph updates.


Due to the large scale of bipartite graphs and the rapid streaming of graph changes, efficiently processing both the snapshot and the continuous CD-SBN is rather challenging. Thus, we propose effective pruning strategies regarding community constraints (e.g., query keywords, community radius, edge supports, and community scores) to significantly reduce the CD-SBN problem space. We also design a hierarchical synopsis to effectively facilitate candidate community search and develop efficient snapshot and continuous algorithms to retrieve or incrementally maintain actual community answers (via proposed synopsis and pruning methods), respectively.

In this paper, we make the following major contributions.

\begin{enumerate}
    \item We formally define the problem of the \textit{community detection over streaming bipartite network} (CD-SBN) containing snapshot and continuous queries in Section~\ref{sec:problem_definition}.
    \item We present an efficient CD-SBN framework for processing queries over streaming bipartite graphs in Section~\ref{sec:solution_framework}.
    \item We design effective community-level pruning strategies to reduce the search space of the CD-SBN problem in Section~\ref{sec:pruning_strategies}.
    \item We propose a novel hierarchical synopsis to facilitate CD-SBN query processing and devise an efficient procedure for incremental graph maintenance in Section~\ref{sec:synopsis_and_graph_maintenance}.
    \item We develop an efficient algorithm with effective synopsis-level pruning strategies to answer the snapshot CD-SBN query, as well as another algorithm to maintain the result set of a continuous CD-SBN query with low computational cost, in Section~\ref{sec:query_processing}.
    \item We demonstrate the efficiency and effectiveness of our CD-SBN processing approach through extensive experiments over synthetic/real-world graphs in Section~\ref{sec:experiments}.
\end{enumerate}

Section \ref{sec:related_work} reviews the related work on community search/ detection on unipartite graphs and static/streaming bipartite graphs. Finally, Section \ref{sec:conclusion} concludes this paper.

\section{Problem Definition}
\label{sec:problem_definition}

Section~\ref{subsec:streaming_bipartite_network} provides the data model for streaming bipartite networks. Section \ref{subsec:basic_pattern} defines the basic units that indicate the relationship between two users. Section~\ref{subsec:k_r_sigma_bitruss_community} introduces the \textit{relationship score} and the $(k,r,\sigma)$-bitruss to measure the cohesiveness of a subgraph. Section~\ref{subsec:community_detection_over_streaming_bipartite_network} formally defines our problem of the \textit{community detection over streaming bipartite networks} (CD-SBN).

\subsection{Streaming Bipartite Networks}
\label{subsec:streaming_bipartite_network}

\noindent {\bf Bipartite Graphs:} We first define an undirected and weighted bipartite graph as follows.

\begin{definition}
\label{def:bipartite_graph}
(\textbf{Bipartite Graph, $G$}) A bipartite graph, $G$, is represented by a quadruple $(U(G), L(G), E(G), \Phi(G))$, where $U(G)$ and $L(G)$ are two disjoint sets (types) of vertices, $E(G)$ is a set of edges, $e_{u, v}$, between vertices $u\in U(G)$ and $v\in L(G)$, and $\Phi(G)$ is a mapping function: $U(G) \times L(G) \rightarrow E(G)$. 

Each vertex $v \in L(G)$ has a set, $v.K$, of keywords, and each edge $e_{u,v} \in E(G)$ is associated with a weight $w_{u,v}$.
\end{definition}

Bipartite graphs have been widely used in many real applications, such as movie recommendations for users~\cite{grujic2008MoviesRecommendationNetworks}, recommending \textit{places of interest} (POIs) to tourists~\cite{lang2022POIRecommendationBased}, or online product shopping recommendation for customers~\cite{chen2022SocialBoostedRecommendation, shan2017AdaptiveKendallsCorrelationa, yoon2020ItemRecommendationPredicting}. 

For simplicity,  this paper considers vertices $u$ in the upper vertex layer $U(G)$ as users in real-world applications, and vertices $v$ in the lower vertex layer $L(G)$ as items that are associated with keyword sets $v.K$ such as movie types, POI features, and product descriptions. Moreover, each edge $e_{u,v}$ from user $u$ to item $v$ is associated with a weight $w_{u,v}$, indicating the user-item interaction frequency such as $\#$ of movie views for movie recommendations, $\#$ of POI visits by users, $\#$ of product purchases for online shopping, and so on.


\noindent {\bf Streaming Bipartite Networks:} Below, we introduce the bipartite network in the streaming scenario.

\begin{definition}
\label{def:streaming_bipartite_network}
(\textbf{Streaming Bipartite Network, $G_t$}) A streaming bipartite network $G_t$ consists of an initial bipartite graph $G_0=(U(G_0), L(G_0), E(G_0),\Phi(G_0))$, and an ordered sequence of update items $S =\{p_1, p_2, \ldots, p_t, \ldots\}$, where each item $p_t = (e_{u,v}, t)$ indicates an update operator of weight $w_{u,v}$ for an inserted edge $e_{u,v}$ arriving at timestamp $t$. A \textit{sliding window}, $W_t$, of size $s$ in a streaming bipartite network $G_t$ contains the most recent $s$ update items $p_{t-s+1}$, $p_{t-s+2}$, $\ldots$, and $p_t$, in $G_t$, where $t$ is the current timestamp. 
\end{definition}



In Definition \ref{def:streaming_bipartite_network}, we consider a sliding window, $W_t$, of $s$ update items over a streaming bipartite network $G_t$.
In reality, streaming updates are interactions between users and items in real-world applications (e.g., users' ratings of movies, access records of POIs, and purchase records).
Each update can be the insertion of a new edge/interaction between a user and an item (or equivalently, the increase of the edge weight $w_{u,v}$). Since the interaction records are time-sensitive, we employ a sliding window, $W_t$, to focus on recently coming edges and ignore those expired edge updates.

After applying update items in the sliding window $W_t$ to the bipartite graph $G$, we can obtain the latest snapshot, $G_t$, of the streaming bipartite network at timestamp $t$.

\noindent {\bf Sliding Window Maintenance of Streaming Bipartite Networks:} We maintain the sliding window $W_t$ upon dynamic updates in 
$G_t$. When a new item $p_{t+1}$ arrives at timestamp $(t+1)$,  we will add $p_{t+1}$ to $W_t$ and remove the expired item $p_{t-s+1}$ from $W_t$, which yields a new sliding window $W_{t+1}$.

\underline{\it Insertion of a New Item $p_{t+1}$.} For the insertion of a new item $p_{t+1}$ $ = (e_{u,v}, t+1)$, we consider two cases:
\begin{itemize}
    \item When both vertices $u$ and $v$ exist in $G_t$ (i.e., $u \in U(G_t)$ and $v \in L(G_t)$),  we add edge $e_{u,v}$ to $G_t$ and increment edge weight $w_{u,v}$ by 1 (i.e., $w_{u,v} = w_{u,v}+1$), and; 
    \item When either of vertices $u$ and $v$ does not exist in $G_t$ (i.e., $u \notin U(G_t)$ and/or $v \notin L(G_t)$), we add new vertex(es) $u$ and/or $v$ to $U(G_t)$ and/or $L(G_t)$, respectively, insert a new edge $e_{u,v}$ into $G_t$, and set edge weight $w_{u,v}$ to 1 (i.e., $w_{u,v} = 1$).
\end{itemize}

\underline{\it Expiration of an Old Item $p_{t-s+1}$.} When an old item $p_{t-s+1} = (e_{u,v}, t-s+1)$ expires, we decrement the edge weight $w_{u,v}$ by 1 (i.e., $w_{u,v}= w_{u,v}-1$). If $w_{u,v}=0$ holds, we remove the edge $e_{u,v}$ from $G_t$ (keeping two ending vertices $u$ and $v$).

\subsection{Basic Patterns in the Bipartite Graph}
\label{subsec:basic_pattern}


In this subsection, we give the definitions of basic patterns, \textit{wedge} and \textit{butterfly} \cite{gou2023SlidingWindowbasedApproximate}, in the bipartite graph $G$. 

\noindent {\bf Wedge:} A wedge is formally defined as follows:

\begin{definition}
\label{def:wedge}
(\textbf{Wedge, $\angle(u_i, v, u_j)$}) Given two user vertices $u_i, u_j$ $\in U(G)$ and an item vertex $v \in L(G)$ in a bipartite graph $G$, a \textit{wedge} \cite{gou2023SlidingWindowbasedApproximate}, $\angle(u_i, v, u_j)$, is a path $u_i \rightarrow v \rightarrow u_j$, where $u_i$, $v$, and $u_j$ are called start-, middle-, and end-vertex of the wedge, respectively. The weight of a wedge $\angle(u_i,v,u_j)$ is given by:
\begin{equation}
\label{eq:wedge_weight}
    w_{\angle(u_i,v,u_j)}=\min\{w_{u_i,v}, w_{u_j,v}\}.
\end{equation}

\end{definition}
\vspace{-1ex}

Based on Definition~\ref{def:wedge}, a wedge $\angle(u_i,v,u_j)$ in the bipartite graph can be two audiences watching the same movie, two travelers visiting the same site, or two customers purchasing the same product. 
The weight of the wedge $\angle(u_i,v,u_j)$ takes the minimum edge weights between $w_{u_i,v}$ and $w_{u_j,v}$, indicating the intensity (or frequency) of \textit{maximally common} interest/behavior $v$ (e.g., movie, POI site, or product) between two users $u_i$ and $u_j$.
In the example of Figure \ref{fig:butterfly_example}, $\angle(u_1,v_1,u_2)$ is one of the wedges, with weight $w_{\angle(u_1,v_1,u_2)} = \min\{2, 4\} = 2$.


\begin{figure}[t]
  \centering
  \includegraphics[width=0.2\textwidth]{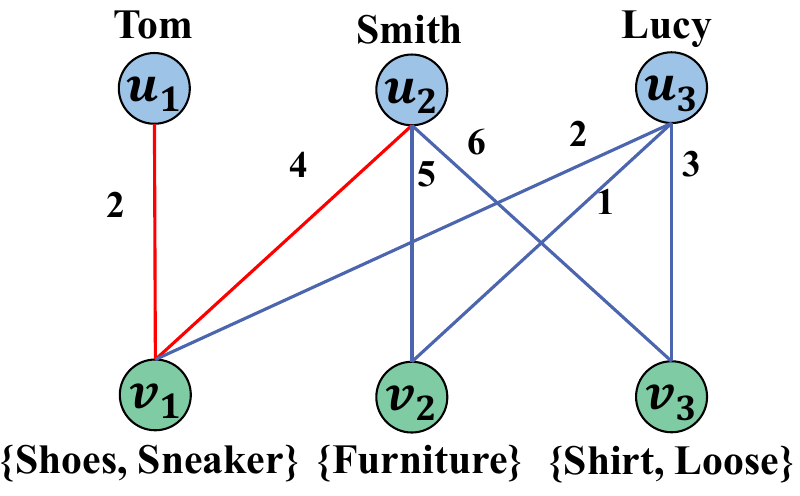}
  \vspace{-3ex}
  \caption{\small An example of basic patterns over a bipartite network.}
  \Description{ An example of basic patterns over a bipartite network. }
  \vspace{-4ex}
  \label{fig:butterfly_example}
\end{figure}

\noindent {\bf Butterfly:} Two wedges form a \textit{butterfly}, as defined below:
\vspace{-1ex}
\begin{definition}
\label{def:butterfly}
(\textbf{Butterfly \cite{gou2023SlidingWindowbasedApproximate}, $\Join(u_i, u_j, v_a, v_b)$}) Given two user vertices $u_i, u_j \in U(G)$ and two item vertices $v_a, v_b \in L(G)$ in a bipartite graph $G$, a \textit{butterfly}, $\Join(u_i, u_j, v_a, v_b)$, is a complete bipartite subgraph \nop{(i.e., $2 \times 2$ biclique)} containing wedges $\angle(u_i, v_a, u_j)$ and $\angle(u_i, v_b, u_j)$.
\end{definition}
\vspace{-1ex}

The butterfly structure (as given in Definition \ref{def:butterfly}) has been widely used in many algorithms and applications \cite{sanei-mehri2018ButterflyCountingBipartite, xu2022EfficientLoadbalancedButterfly, zhou2021ButterflyCountingUncertain} to detect communities with high cohesiveness in bipartite graphs. In Figure \ref{fig:butterfly_example}, $\Join(u_2,u_3,v_2,v_3)$ is an example of a butterfly.



\noindent {\bf The Score of a Butterfly:} 
The butterfly is a basic graphlet pattern of bipartite graphs. Below, we define a \textit{butterfly score} to measure the relationship among user and item vertices in the butterfly.


\vspace{-1ex}
\begin{definition}
\label{def:weighted_butterfly_counting}
(\textbf{The Butterfly Score}) The score, $w_{\Join(u_i, u_j, v_a, v_b)}$, of a butterfly $\Join(u_i, u_j, v_a, v_b)$ is given by:
\begin{eqnarray}
\label{eq:butterfly_score}
    \vspace{-3ex}
    w_{\Join(u_i, u_j, v_a, v_b)} &=& w_{\angle(u_i,v_a,u_j)} \cdot w_{\angle(u_i,v_b,u_j)}\\
    &=& \min\{w_{u_i,v_a}, w_{u_j,v_a}\} \cdot \min\{w_{u_i,v_b}, w_{u_j,v_b}\}.\notag
\end{eqnarray}
\end{definition}

In Definition \ref{def:weighted_butterfly_counting}, we consider two wedges $\angle(u_i,v_a,u_j)$ and $\angle(u_i,v_b,$ $u_j)$ in the butterfly $\Join(u_i, u_j, v_a, v_b)$, and define the score of the butterfly by multiplying their corresponding weights, that is, $\min\{w_{u_i,v_a},$ $w_{u_j,v_a}\} \cdot \min\{w_{u_i,v_b},$ $w_{u_j,v_b}\}$. Intuitively, we treat the weight of each wedge as the frequency of duplicate wedges, and the score of the butterfly is given by the count of duplicate butterflies (i.e., the multiplication of these two wedge weights). 

For the previous example in Figure \ref{fig:butterfly_example}, the score, $w_{\Join(u_2,u_3,v_2,v_3)}$, of the butterfly $\Join(u_2,u_3,v_2,v_3)$ can be calculated by: $\min\{5, 1\}\cdot\min\{6,3\} = 1 \times 3 = 3$.


\vspace{-1ex}
\subsection{The \texorpdfstring{$(k,r,\sigma)$}{(k,r,s)}-Bitruss Community in a Bipartite Graph}
\label{subsec:k_r_sigma_bitruss_community}

In this subsection, we first define the \textit{user relationship score} between any two user vertices $u_i$ and $u_j$ in a bipartite graph $G$ and then formalize the $(k,r,\sigma)$-bitruss community with high structural cohesiveness in $G$. 

\noindent {\bf The Relationship Score of a User Vertex Pair:} Based on the score of a butterfly (as given in Eq.~(\ref{eq:butterfly_score})), we can define the \textit{user relationship score} of a user vertex pair  $(u_i, u_j,u_j)$ as follows.

\begin{definition}
\label{def:user_relationship_score}
(\textbf{User Relationship Score, $score_{u_i,u_j}$}) Given two vertices $u_i, u_j \in U(G_t)$ in a bipartite graph $G_t$, the \textit{user relationship score}, $score_{u_i, u_j}$, between $u_i$ and $u_j$ in $G_t$ can be calculated by:
\vspace{-1ex}
\begin{equation}
\label{eq:user_relationship_score}
    score_{u_i, u_j}(G_t) = \sum_{\forall v_a,v_b \in L(G_t)}{w_{\Join(u_i,u_j,v_a,v_b)}}.
\end{equation}
\end{definition}
\vspace{-1ex}

In Definition \ref{def:user_relationship_score}, the user relationship score $score_{u_i, u_j}$ between two users $u_i$ and $u_j$ is defined as the summed butterfly score for all butterflies containing vertices $u_i$ and $u_j$. A higher relationship score indicates a stronger relationship between users $u_i$ and $u_j$.

As illustrated in Figure \ref{fig:butterfly_example}, the relationship score, $score_{u_2, u_3}$, between users $u_2$ and $u_3$ is given by: $w_{\Join(u_2,u_3,v_1,v_2)} + w_{\Join(u_2,u_3,v_1,v_3)} + w_{\Join(u_2,u_3,v_2,v_3)} = \min\{4,2\} \cdot \min\{5,1\} + \min\{4,2\} \cdot \min\{6,3\} + \min\{5,1\} \cdot \min\{6,3\} = 2 \times 1 + 2 \times 3 + 1 \times 3 = 11$.


\noindent {\bf The $\mathbf{(k,r,\sigma)}$-Bitruss:} We next define a $(k,r,\sigma)$-bitruss community in the bipartite graph $G$ with high structural and user relationship cohesiveness. 

\vspace{-2ex}
\begin{definition}
\label{def:k_r_sigma_bitruss}
 (\textbf{$(k,r,\sigma)$-Bitruss}) Given a bipartite graph $G_t=(U(G_t),L(G_t),E(G_t),\Phi(G_t))$, a center vertex $u_c \in U(G_t)$, a threshold $k$ of butterfly number, a maximum radius $r$, and a threshold, $\sigma$, of the user relationship score, a \textit{$(k,r,\sigma)$-bitruss} is a connected subgraph, $g$, of $G_t$ (denoted as $g \subseteq G_t$), such that:
 \vspace{-1ex}
 \begin{itemize}
    \item (Support) for each edge $e_{u, v} \in E(g)$, the edge support $sup(e_{u,v})$ (defined as the number of butterflies containing edge $e_{u, v}$) is not smaller than $k$;
    \item (Radius) for any user vertex $u_i \in U(g)$, we have $dist(u_c, u_i) \leq 2r$, and;
    \item (Score) for any two user vertices $u_i, u_j \in \angle(u_i, v, u_j)$, we have $score_{u_i, u_j}(g) \geq \sigma$,
 \end{itemize}
where $u_c \in g$, and $dist(u_c, u_i)$ is the shortest path distance between $u_c$ and $u_i$ in bipartite subgraph $g$.
\end{definition}

\vspace{-1ex}
The $(k,r,\sigma)$-bitruss $g$ in Definition~\ref{def:k_r_sigma_bitruss} is defined as a connected subgraph that: 1) each edge is contained in at least $k$ butterflies, 2) subgraph $g$ is centered at $u_c$ with a maximum radius $2r$, and 3) each pair of user vertices $u_i, u_j$ share a common neighbor $v$ holding $score_{u_i, u_j}(g) \geq \sigma$. Intuitively, in real applications, the $(k,r,\sigma)$-bitruss refers to a user group with strong connections (measured by a user relationship score), for example, users who enjoy movies of the same genre, visitors who visited the same POIs before, or customers with many common purchased products.


\vspace{-1ex}
\subsection{The Problem of the Community Detection Over Streaming Bipartite Network (CD-SBN)}
\label{subsec:community_detection_over_streaming_bipartite_network}

\noindent {\bf The CD-SBN Problem:} We are now ready to define the problem of community detection over streaming bipartite network.

\vspace{-1ex}
\begin{definition}
\label{def:cd_sbn_problem}
(\textbf{Community Detection Over Streaming Bipartite Network, CD-SBN}) Given a streaming bipartite network $G_t$, a support threshold $k$, a maximum radius $r$, a threshold, $\sigma$, of the user relationship score, and a set, $Q$, of query keywords, the problem of the \textit{community detection over streaming bipartite network} (CD-SBN) retrieves a result set, $R$, of subgraphs $g_i$ of $G_t$ (i.e., $g_i \subseteq G_t$), such that:
\begin{itemize}
    \item (Keyword Relevance) for any item vertex $v_i \in L(g_i)$, its keyword set $v_i.K$ must contain at least one query keyword in $Q$ (i.e., $v_i.K \cap Q \neq \emptyset$), and;
    \item (Structural and User Relationship Cohesiveness) bipartite subgraph $g_i$ is a $(k,r,\sigma)$-bitruss (as given in Definition \ref{def:k_r_sigma_bitruss}).
\end{itemize}
\end{definition}

Note that the CD-SBN problem has two scenarios: \textit{one-time snapshot} and \textit{continuous CD-SBN}. In particular, the one-time snapshot CD-SBN problem finds CD-SBN communities over a snapshot, $G_t$, of a streaming bipartite graph at timestamp $t$, whereas the continuous CD-SBN continuously monitors CD-SBN communities upon edge updates (i.e., $S$) in $G_t$, for some registered community constraints in Definition \ref{def:cd_sbn_problem}. The former version can be used to obtain initial CD-SBN community answers when CD-SBN predicates are online-specified, and the latter dynamically maintains CD-SBN community results upon graph updates.

\noindent {\bf Challenges:} The processing of snapshot and continuous CD-SBN queries is challenging.
One straightforward method to tackle the snapshot problem is to enumerate all community candidates and check whether they satisfy the CD-SBN constraints. However, this method is quite inefficient because of the costly checking of CD-SBN community predicates, including how to efficiently verify the constraints of keywords, graph structural cohesiveness, and user relationship scores. For the continuous CD-SBN problem, a straightforward method is to answer the snapshot query at each timestamp $t$ upon graph updates over a sliding window $W_t$. Nonetheless, there are many possible CD-SBN community candidates, which incur high computation costs to enumerate and refine. Therefore, it is also challenging to incrementally maintain the dynamically changing graph $G_t$ and efficiently update the CD-SBN query answers.


Table~\ref{table:symbols_and_descriptions} lists the commonly used symbols and their descriptions in this paper.

\setlength{\textfloatsep}{0pt}
\begin{table}[t]
\caption{Symbols and descriptions.}\small
\vspace{-4ex}
\footnotesize
\label{table:symbols_and_descriptions}
\begin{center}
\begin{tabular}{|l|p{6cm}|}
\hline
\textbf{Symbol}&{\textbf{Description}} \\
\hline\hline
$G$ & a bipartite graph \\
\hline
$G_t$ & a streaming bipartite network at timestamp $t$\\
\hline
$U(G)$ & a set of user vertices $u$ in the upper layer of bipartite graph $G$\\
\hline
$L(G)$ & a set of item vertices $v$ in the lower layer of bipartite graph $G$\\
\hline
$E(G)$ & a set of edges $e_{u,v}$ for $u\in U(G)$ and $v\in L(G)$\\
\hline
$S$ & an ordered sequence of update items $p_t = (e_{u, v}, t)$\\
\hline
$u$ (or $u_i$, $u_j$) & a user vertex in $U(G)$ of bipartite graph $G$\\
\hline
$v$ (or $v_a$, $v_b$) & an item vertex in $L(G)$ of bipartite graph $G$\\
\hline
$W_t$ & a sliding window of size $s$ at timestamp $t$\\
\hline
$w_{u,v}$& the edge weight of edge $e_{u,v}$\\
\hline
$\angle (u_i, v, u_j)$ & a wedge with $u_i, u_j \in U(G)$ and $v\in L(G)$\\
\hline
$\Join(u_i, u_j, v_a, v_b)$ & a butterfly with $u_i, u_j \in U(G)$ and $v_a, v_b \in L(G)$\\
\hline
$w_{\Join(u_i, u_j, v_a, v_b)}$ & the score of a butterfly $\Join(u_i, u_j, v_a, v_b)$\\
\hline
$k$ & the support threshold for a $(k,r,\sigma)$-bitruss\\
\hline
$sup(e)$ & the number of butterflies containing edge $e$\\
\hline
$r$ & the maximum radius for a $(k,r,\sigma)$-bitruss\\
\hline
$\sigma$ & the relationship score threshold for a $(k,r,\sigma)$-bitruss\\
\hline
$v_i.K$ & the keyword set for each vertex $v_i \in L(G)$ \\
\hline
\end{tabular}
\label{tab1}
\end{center}
\vspace{-1ex}
\end{table}
\setlength{\textfloatsep}{0pt}

\section{The CD-SBN Processing Framework}
\label{sec:solution_framework}

Figure~\ref{algo:the_solution_framework} illustrates a framework for answering snapshot/continuous CD-SBN queries, which consists of three components, \textit{initialization}, \textit{graph incremental maintenance}, and \textit{CD-SBN query processing}.

To support the CD-SBN stream processing, the \textit{initialization component} prepares some auxiliary data over the initial bipartite graph $G_0$ at timestamp $0$, including a keyword bit vector $v.BV$ for hashing keywords in item vertices, wedge-score-related data for the butterfly score calculation, support upper bound $ub\_sup(e_{u,v})$, and a synopsis $Syn$ (lines 1-4).


At each timestamp $t$, the sliding window $W_t$ includes a new item $p_t$ (i.e., the edge insertion) and removes an expired item $p_{t-s}$ (i.e., the edge deletion). Therefore, a \textit{graph incremental maintenance component} will compute an incremental update of wedge-score-related data and obtain the updated wedge-score-related data at timestamp $t$ (lines 5-6). Similarly, for those affected edges $e'$ (due to edge insertion/deletion), we will update support upper bounds $ub\_sup(e')$ (line 7) and synopsis $Syn$ (line 8).


For the \textit{CD-SBN query processing component}, the snapshot CD-SBN algorithm traverses the synopsis $Syn$ to retrieve candidate communities $g$ by applying our proposed pruning strategies, and refines the resulting candidate communities $g$ to obtain actual CD-SBN answers in $R$ (lines 8-11). Moreover, the continuous CD-SBN algorithm incrementally maintains CD-SBN community answers in $R$ by removing communities with expired edges and inserting new CD-SBN communities due to edge insertion (lines 12-14).


\setlength{\textfloatsep}{0pt}
\begin{algorithm}[t]
\caption{{\bf The CD-SBN Query Process Framework}}\small
\label{algo:the_solution_framework}
\KwIn{
    \romannumeral1) a streaming bipartite network $G_t$;
    \romannumeral2) a set, $Q$, of query keywords;
    \romannumeral3) a support threshold $k$;
    \romannumeral4) a maximum radius $r$;
    \romannumeral5) a threshold, $\sigma$, of user relationship score

}
\KwOut{
    a set, $R$, of CD-SBN communities
}

\tcp{\bf Initialization ($t=0$)}
hash keywords in $v.K$ into a keyword bit vector $v.BV$ for each item vertex $v \in L(G_0)$\\
calculate the wedge-score-related data for each pair of user vertices $u_i, u_j \in U(G_0)$
compute the support upper bound $ub\_sup(e_{u,v})$ for each edge $e_{u,v} \in E(G_0)$\\
construct a hierarchical synopsis, $Syn$, over $G_0$\\

\For{each timestamp $t$ ($>0$)}{
    \tcp{\bf Graph Incremental Maintenance}
    compute the increment of wedge-score-related data and update the data for each pair of user vertices $u_i, u_j \in U(G_t)$\\
    update support upper bounds $ub\_sup(e')$ for relevant edges $e'$ (upon edge insertion of $p_t$ and deletion of $p_{t-s}$) \\
    update synopsis $Syn$\\
    
    \tcp{\bf CD-SBN Query Processing}
    \For{each snapshot CD-SBN query}{
       traverse the synopsis $Syn$ by applying \textit{keyword}, \textit{support}, and \textit{score upper bound pruning strategies} to retrieve candidate communities $g$\\
       refine candidate communities $g$ and obtain actual CD-SBN results $R$\\
       \Return $R$\\
    }
    
    \For{each continuous CD-SBN query}{
        remove the expired communities from $R$\\
        insert new CD-SBN communities into $R$\\
    }
}
\end{algorithm}
\setlength{\textfloatsep}{0pt}

\section{Pruning Strategies}
\label{sec:pruning_strategies}
This section outlines the community-level pruning strategies (see Section~\ref{sec:solution_framework}) for reducing the search space of the CD-SBN problem. 
\nop{
{\bf Due to space limitations, we omit proofs of all the lemmas throughout this paper. Please refer to the complete proofs in our technical report on Arxiv~\cite{zhang2024EffectiveCommunityDetection}.}
}

\subsection{Keyword Pruning}
\label{subsec:keyword_pruning}
Based on the CD-SBN problem definition (Definition~\ref{def:cd_sbn_problem}), any item vertex $v_i$ in the lower layer of a candidate bipartite subgraph $g$ should contain at least one query keyword in $Q$. In the sequel, we present an effective \textit{keyword pruning}, which filters out item vertices $v_i$ in subgraph $g$ that do not contain any query keywords in $Q$.

\begin{lemma}
\label{lemma:keyword_pruning}
    \textbf{(Keyword Pruning)} Given a query keyword set $Q$ and a candidate bipartite subgraph $g$, any item vertex $v_i \in L(g)$ can be safely pruned from the subgraph $g$, if $v_i.K \cap Q = \emptyset$ holds, where $v_i.K$ is the keyword set associated with item vertex $v$.
\end{lemma}

\begin{proof}
\label{proof:keyword_pruning}
    If $v_i.K \cap Q = \emptyset$ holds for an item vertex $v_i \in L(g)$, then vertex $v_i$ does not contain any query keyword in $Q$, which violates the keyword relevance constraint (as given in Definition~\ref{def:cd_sbn_problem}). Hence, the item vertex $v_i$ cannot be in the CD-SBN subgraph answer and can thus be safely pruned, which completes the proof.
\end{proof}
\nop{
\begin{proof}
Please refer to the proof in our technical report~\cite{zhang2024EffectiveCommunityDetection}.
\end{proof}
}

\subsection{Support Pruning}
\label{subsec:support_pruning}
As mentioned in Definition~\ref{def:cd_sbn_problem}, since the CD-SBN community should be a $(k,r,\sigma)$-bitruss, the support, $sup(e_{u,v})$, of edge $e_{u,v}$ should not be smaller than $k$ (i.e., $sup(e_{u,v}) \geq k$). Assume that we can obtain an upper bound, $ub\_sup(e_{u,v})$, of support $sup(e_{u,v})$ for an edge $e_{u,v} \in E(g)$. We provide the following  \textit{support pruning} lemma to rule out edges in a candidate bipartite subgraph $g$ with low supports.

\begin{lemma}
\label{lemma:support_pruning}
    \textbf{(Support Pruning)} Given a support threshold $k$ and a candidate bipartite subgraph $g$, the edge $e_{u,v} \in E(g)$ can be safely pruned from subgraph $g$ if it satisfies $ub\_sup(e_{u,v}) < k$.
\end{lemma}

\begin{proof}
\label{proof:support_pruning}
    According to Definition~\ref{def:k_r_sigma_bitruss}, the support, $sup(e_{u,v})$, is defined as the number of butterflies containing a specific edge. Thus, each edge $e$ must be reinforced by at least $k$ butterflies in a $(k,r,\sigma)$-bitruss. Since $ub\_sup(e_{u,v})$ is an upper bound of support $sup(e_{u,v})$, we have $sup(e_{u,v})\leq ub\_sup(e_{u,v})$. Moreover, from the assumption of the lemma that $ub\_sup(e_{u,v}) < k$ holds, by the inequality transition, we have $sup(e_{u,v})\leq ub\_sup(e_{u,v}) < k$. Hence, the support $sup(e_{u,v})$ of the edge $e_{u,v}$ is below the threshold $k$, and can be safely pruned, which completes the proof.
\end{proof}
\nop{
\begin{proof}
Please refer to the proof in our technical report~\cite{zhang2024EffectiveCommunityDetection}.
\end{proof}
}

\noindent {\bf Discussions on How to Obtain the Upper Bound of Support $\mathbf{ub\_sup(e_{u,v})}$.}
In the offline phase, the upper bound on support must be computed to enable support pruning. Since part of the item vertices in the data graph $G_t$ will be filtered out by the \textit{keyword pruning}, where the rest of the vertices induce a subgraph $g$, the number of butterflies containing the edge $e_{u,v}$ in $g$ is smaller than or equal to that in $G_t$. Thus, the support of $e_{u,v}$ in $G_t$ can be considered as an upper bound $ub\_sup(e_{u,v})$ in the subgraph $g$ defined below. 
\begin{equation}
\label{eq:edge_support_upper_bound}
    ub\_sup(e_{u,v}) = \sum_{\forall v_x\in N(u)-\{v\}}{|N(v)\cap N(v_x)-\{u\}|},
\end{equation}
where $N(v)$ is the neighbor vertices set of vertex $v$. With $u\in U(G_t)$ and $v\in L(g)$, the number of butterflies containing $e_{u,v}$ in $G_t$ can be calculated as the sum of the number of common neighbors for $v$ and each neighbor $v_x$ of $u$.

\subsection{Layer Size Pruning}
\label{subsec:layer_size_pruning}


Since the support of a bipartite subgraph is closely related to the subgraph size, we can filter out a candidate community $g$ with small sizes of upper/lower layers (i.e., $|U(g)|$ and $|L(g)|$), with respect to the support threshold $k$.

\begin{lemma}
\label{lemma:layer_size_pruning}
    \textbf{(Layer Size Pruning)} Given a support threshold $k$ and a candidate bipartite subgraph $g=(U(g), L(g), E(g), \Phi(g))$, subgraph $g$ can be safely pruned, if it satisfies the condition that $(|U(g)|-1)\cdot(|L(g)|-1) < k$.
\end{lemma}

\begin{proof}
\label{proof:layer_size_pruning}
    According to Definition~\ref{def:k_r_sigma_bitruss}, the support of each edge $e$ in $(k,r,\sigma)-bitruss$ must be higher than the support threshold $k$, which means that each edge $e$ must be contained by at least $k$ butterflies. For each edge, $e_{u_i, v_a}$, any pair of other user vertex and item vertex, $(u_x, v_y)$, can form a butterfly, $\Join(u_i,u_x,v_a,v_y)$, with the two ending vertices of edge $e_{u_i, v_a}$. Assume that $g$ is a fully connected bipartite graph, the edge $e_{u_i,v_a}$ is supported by pairs of vertices $|U-\{u_i\}|\times|L-\{v_a\}|$ so that its support is $sup(e_{u_i,v_a})=(|U(g)|-1)\cdot(|L(g)|-1)$. However, not every pair of vertices is connected to $u_i$ or $v_a$ in general, hence it holds that $sup(e_{u_i,v_a}) \leq (|U(g)|-1)\cdot(|L(g)|-1)$. If $(|U(g)|-1)\cdot(|L(g)|-1)<k$ holds, we have $sup(e_{u_i,v_a}) \leq (|U(g)|-1)\cdot(|L(g)|-1) < k$. Thus, edge $e_{u_i,v_a}$ can be safely pruned, which completes the proof.
\end{proof}

\subsection{Radius Pruning}
\label{subsec:radius_pruning}
In the definition of $(k,r,\sigma)$-bitruss, the radius parameter $r$ ensures that the distances between the center vertex and other user vertices are smaller than or equal to $2r$. Thus, we devise a radius pruning strategy to discard a user vertex $u_i$ whose distance, $dist(u_i, u_q)$, to the center vertex $u_q$ is greater than $2r$.

\begin{lemma}
\label{lemma:radius_pruning}
    \textbf{(Radius Pruning)} Given a candidate bipartite subgraph $g$ (centered at user vertex $u_q$) and a radius $r$, a user vertex $u_i \in U(g)$ can be safely pruned from subgraph $g$, if it holds $dist(u_i, u_q) > 2r$.
\end{lemma}

\begin{proof}
\label{proof:radius_pruning}
    For the candidate subgraph $g$ centered at the user vertex $u_q$, if $dist(u_i, u_q) > 2r$ holds for a user vertex $u_i \in U(g)$, then $u_i$ violates the radius constraint. Therefore, we should prune the user vertex $u_i$ from $g$, which completes the proof.
\end{proof}

\nop{
\begin{proof}
Please refer to the proof in our technical report~\cite{zhang2024EffectiveCommunityDetection}.
\end{proof}
}

\subsection{Score Upper Bound Pruning}
\label{subsec:score_upper_bound_pruning}

Based on Definition~\ref{def:k_r_sigma_bitruss}, any pair of user vertices belonging to a wedge in a $(k,r,\sigma)$-bitruss must have the user relationship score greater than or equal to $\sigma$, where $\sigma$ is a score threshold specified by users online. In the sequel, we propose a \textit{score upper bound pruning} strategy to eliminate the user vertex(es) with low scores.

\begin{lemma}
\label{lemma:score_upper_bound_pruning}
    \textbf{(Score Upper Bound Pruning)} Given a community $g$ and a candidate user vertex $u_i\in U(g)$, the user vertex $u_i$ can be safely pruned, if there exists a user vertex $u_j \in U(g)$ and a wedge $\angle(u_i, v, u_j) \subseteq g$, such that $ub\_score_{u_i, u_j}(g) < \sigma$ holds, where $ub\_score_{u_i, u_j}(g)$ is an upper bound of score $score_{u_i, u_j}(g)$.
\end{lemma}

\nop{
\begin{proof}
\label{proof:score_upper_bound_pruning}
Given a candidate user vertex $u_i$ and a user vertex $u_j \in U(g)$, since $ub\_score_{u_i, u_j}(g)$ is a score upper bound, it holds that $score_{u_i, u_j}(g) \leq ub\_score_{u_i, u_j}(g)$. From the assumption of the lemma, if $ub\_score_{u_i, u_j}(g) < \sigma$ holds, by the inequality transition, we have $score_{u_i, u_j}(g)\leq ub\_score_{u_i, u_j}(g) < \sigma$, violating the score constraint of the $(k,r,\sigma)$-bitruss (given in Definition \ref{def:k_r_sigma_bitruss}). Thus, the candidate user vertex $u_i$ can be safely pruned, which completes the proof.
\end{proof}

\begin{proof}
Please refer to the proof in our technical report~\cite{zhang2024EffectiveCommunityDetection}.
\end{proof}
}

\noindent {\bf Discussions on How to Obtain the Score Upper Bound,\\ $ub\_score_{u_i, u_j}(g)$:} To enable the score upper bound pruning, we need to calculate the score upper bound, $ub\_score_{u_i,u_j}(g)$, for the user pair $(u_i, u_j)$ in subgraph $g$ that is connected to some common-neighbor item vertices with the desired (query) keywords. While Eq.~(\ref{eq:user_relationship_score}) computes the user relationship score $score_{u_i, u_j}(G_t)$, by considering all possible common-neighbor item vertices $v\in L(G_t)$ of users $u_i$ and $u_j$, and ignoring the constraint of query keywords (i.e., $Q$), the score $score_{u_i,u_j}(G_t)$ is essentially an upper bound of $score_{u_i,u_j}(g)$ (since $g\subseteq G$ holds). Therefore, we can use $score_{u_i,u_j}(G_t)$ as the score upper bound $ub\_score_{u_i, u_j}(g)$ (i.e., $ub\_score_{u_i, u_j}(g) = score_{u_i,u_j}(G_t)$).

\noindent {\bf Computation of the Score $score_{u_i, u_j}(G_t)$:} Now the only remaining issue is how to efficiently compute the user relationship score, $score_{u_i, u_j}(G_t)$, which is given by the summed score for all butterflies containing $u_i$ and $u_j$ by combining Eqs.~(\ref{eq:user_relationship_score}) and ~(\ref{eq:butterfly_score}).

Based on the \textit{multinomial theorem} \cite{bolton1968MultinomialTheorem} with power $n=2$, we have:
\vspace{-4ex}
\begin{eqnarray}
    \left(\sum_{i=1}^m {o_i}\right)^2 &=& \sum_{i=1}^m {o_i}^2 + 2\sum_{1\leq i < j \leq m}{o_i o_j},\notag
\end{eqnarray}
which can be rewritten as:
\vspace{-3ex}
\begin{eqnarray}
    \sum_{1\leq i <j \leq m}{o_i o_j} &=& \frac{1}{2}\left(\left(\sum_{i=1}^m {o_i}\right)^2 - \sum_{i=1}^m {o_i}^2\right).
    \label{eq:multinomial_theorem}
\end{eqnarray}

Based on Eq.~(\ref{eq:multinomial_theorem}), the user relationship score, $score_{u_i, u_j}(G_t)$, between $u_i$ and $u_j$ (as given in Eq.~(\ref{eq:user_relationship_score})) can be rewritten as:
\begin{eqnarray}
\label{eq:user_relationship_score_expansion}
    &&score_{u_i, u_j}(G_t) = \sum_{\forall a<b, v_a,v_b \in L(G_t)}{w_{\Join(u_i,u_j,v_a,v_b)}} \\
    &=& \sum_{\forall a<b, v_a,v_b \in L(G_t)}{w_{\angle(u_i, v_a, u_j)} \cdot {w_{\angle(u_i, v_b, u_j)}}}\notag\\
    &=& \frac{1}{2}\left(\left(\sum_{\forall v \in L(G_t)} {w_{\angle(u_i, v, u_j)}}\right)^2-\sum_{\forall v \in L(G_t)}\left(w_{\angle(u_i, v, u_j)}\right)^2\right)\notag\\
    &=& \frac{1}{2}\left(\left(\sum_{\forall v \in L(G_t)} \hspace{-2ex}{\min\{w_{u_i,v}, w_{u_j,v}\}}\right)^2\hspace{-1ex}-\hspace{-2ex}\sum_{\forall v \in L(G_t)}\hspace{-2ex}\left(\min\{w_{u_i,v}, w_{u_j,v}\}\right)^2\right)\notag\\
    &=& \frac{1}{2}\left(X_{u_i,u_j}^2 - Y_{u_i,u_j} \right).\notag
\end{eqnarray}

From Eq.~(\ref{eq:user_relationship_score_expansion}), in order to calculate the user relationship score $score_{u_i, u_j}(G_t)$ for each user pair $u_i$ and $u_j$, we only need to maintain two terms, $X_{u_i,u_j} = \sum_{\forall v \in L(G_t)}{\min\{w_{u_i,v}, w_{u_j,v}\}}$ and $Y_{u_i,u_j} = \sum_{\forall v \in L(G_t)}\left(\min\{w_{u_i,v}, w_{u_j,v}\}\right)^2$. Upon graph updates in stream $S$, we can incrementally and efficiently update $X_{u_i,u_j}$ and $Y_{u_i,u_j}$ online (as will be discussed later in Section~\ref{subsec:incremental_maintenance_graph}).



\vspace{-1ex}
\section{Auxiliary Data Initialization and Incremental Maintenance}
\label{sec:synopsis_and_graph_maintenance}
This section details the initialization stage (lines 1-4 of Algorithm~\ref{algo:the_solution_framework}) that computes auxiliary data over an initial bipartite graph $G_0$ in Section~\ref{subsec:data_pre_computation}, constructs a synopsis, $Syn$, for such auxiliary data in Section~\ref{subsec:synopsis_construction}, and incremental maintenance of auxiliary data and synopsis (lines 5-8 of Algorithm~\ref{algo:the_solution_framework}) in Sections~\ref{subsec:incremental_maintenance_graph} and ~\ref{subsec:incremental_maintenance_synopsis}, respectively.

\vspace{-2ex}
\subsection{Auxiliary Data Initialization}
\label{subsec:data_pre_computation}
To facilitate efficient online community detection, we will compute auxiliary data over the initial graph $G_0$ during the initialization phase, which will enable our proposed pruning strategies (as mentioned in Section~\ref{sec:pruning_strategies}).
\vspace{-1ex}
\begin{itemize}
    \item {\bf (Item Keyword Bit Vector)} a keyword bit vector, $v.BV$, of size $B$ for each item vertex $v \in L(G_0)$, whose elements $v.BV[i]$ are $1$ if a keyword in $v.K$ is hashed to the $i$-th position (otherwise, $v.BV[i]=0$);
    \item {\bf (Support Upper Bound)} an upper bound, $ub\_sup(e_{u,v})$ (as given by Eq.~(\ref{eq:edge_support_upper_bound})), of the support $sup(e_{u,v})$, for each edge $e_{u,v} \in E(G_0)$, and;
    \item {\bf (Aggregated Wedge Scores)} the summed wedge score, $X_{u_i,u_j}$ $ =\sum_{\forall v \in L(G_t)}$ $w_{\angle(u_i, v, u_j)}$, and the summation over the squares of wedge scores, $Y_{u_i,u_j} = $ $\sum_{\forall v \in L(G_t)} (w_{\angle(u_i, v, u_j)})^2$,  for each user pair $(u_i,u_j) \in U(G_0)\times U(G_0)$ and item vertices $v \in N(u_i)\cap N(u_j)$.
\end{itemize}


\vspace{-1ex}
\subsection{Incremental Maintenance of Auxiliary Data}
\label{subsec:incremental_maintenance_graph}
In this subsection, we illustrate how to incrementally maintain auxiliary data for each user vertex in graph $ G_t$ upon changes to the update items in the sliding window $W_t$. As mentioned in Section~\ref{subsec:streaming_bipartite_network}, the sliding window $W_t$ at a new timestamp $t$ (from $W_{t-1}$) needs to insert a new item $p_t$ and remove an expired old item $p_{t-s}$. Upon such updates, we need to incrementally compute the edge weights in the graph $G_t$, as well as the auxiliary data in the synopsis $Syn$.



\noindent {\bf Maintenance of Support Upper Bounds:} At timestamp $t$, due to the insertion of $p_t$ or the deletion of $p_{t-s}$, an edge $e_{u_i, v_a}$ can be added or deleted, respectively. Consequently, we need to update the support upper bounds of the edges in all butterflies $\Join(u_i, u_j, v_a, v_b)$ containing the edge $e_{u_i, v_a}$. When inserting a new edge $e_{u_i, v_a}$, for each edge $e$ in $\Join(u_i, u_j, v_a, v_b)$, we increase the support upper bound $ub\_sup(e)$ by 1. Similarly, when deleting an edge $e_{u_i, v_a}$, for each edge $e$ in $\Join(u_i, u_j, v_a, v_b)$, we decrease the support upper bound $ub\_sup(e)$ by 1.

 \noindent {\bf Maintenance of User Relationship Scores:} Since the user relationship score $score_{u_i,u_j}(G_t)$ is calculated using $X_{u_i, u_j}$ and $Y_{u_i, u_j}$ (as given by Eq.~(\ref{eq:user_relationship_score_expansion})), we maintain these two terms incrementally. Specifically, we compute the differences, $\Delta X_{u_i,u_j}$ and $\Delta Y_{u_i,u_j}$, of $X_{u_i,u_j}$ and $Y_{u_i,u_j}$, respectively, between timestamps $(t-1)$ and $t$. Upon updating the edge $e_{u_i, v_a}$ at timestamp $t$, the edge weight $w_{u_i, v_a}$ has an increase/decrease of 1. Accordingly, the wedge weight $w_{\angle(u_i, v_a, u_j)}$ ($= min\{w_{u_i, v_a}, w_{u_j, v_a}\}$) is updated with $w_{\angle(u_i, v_a, u_j)}$ $+$ $ \lambda_{\angle(u_i, v_a, u_j)}$, where $\lambda_{\angle(u_i, v_a, u_j)}$$\in$$\{-1, 0, 1\}$ (depending on the change of $min\{w_{u_i, v_a}, w_{u_j, v_a}\}$).

Here, there are three possible cases for the $\lambda_{\angle(u_i, v_a, u_j)}$ value:
\begin{itemize}
    \item {\bf (Case 1)} if the weight $w_{u_i, v_a}$ of edge $e_{u_i, v_a}$ is increased by 1 and $w_{u_i, v_a} + 1 \leq w_{u_j, v_a}$, then we have $\lambda_{\angle(u_i, v_a, u_j)} = 1$;
    \item {\bf (Case 2)} if the weight $w_{u_i, v_a}$ of edge $e_{u_i, v_a}$ is decreased by 1 and $w_{u_i, v_a} - 1 < w_{u_j, v_a}$, then we have $\lambda_{\angle(u_i, v_a, u_j)} = -1$;
    \item {\bf (Case 3)} for the remaining case, we have $\lambda_{\angle(u_i, v_a, u_j)} = 0$.
\end{itemize}

\underline{\it The Computation of Terms $\Delta X_{u_i,u_j}$ and $\Delta Y_{u_i,u_j}$:} We provide the formulae of the two terms $\Delta X_{u_i,u_j}$ and $\Delta Y_{u_i,u_j}$ below.

\begin{eqnarray}
\label{eq:summed_wedge_score_increment}
    && \hspace{-2ex}\Delta X_{u_i,u_j} \\
    &\hspace{-2ex}=& \hspace{-2ex}X_{u_i,u_j}^{(t)}-X_{u_i,u_j}^{(t-1)}\notag\\
    &\hspace{-2ex}=& \hspace{-2ex}\left(\sum_{\forall v \in L(G), v \neq v_a}\hspace{-4ex}{w_{\angle(u_i, v, u_j)}} + (w_{\angle(u_i, v_a, u_j)} + \lambda_{\angle(u_i, v_a, u_j)})\right) - \hspace{-2ex}\sum_{\forall v \in L(G)}\hspace{-2ex}{w_{\angle(u_i, v, u_j)}}\hspace{-4ex} \notag\\
    &\hspace{-2ex}=& \hspace{-2ex}\lambda_{\angle(u_i, v_a, u_j)}, \notag
\end{eqnarray}
where $X_{u_i,u_j}^{(t)}$ and $X_{u_i,u_j}^{(t-1)}$ are the wedge-score-related data $X_{u_i,u_j}$ at timestamps $t$ and $(t-1)$, respectively.

\begin{eqnarray}
\label{eq:summed_square_of_wedge_score_increment}
    && \hspace{-2ex}\Delta Y_{u_i,u_j} \\
    &\hspace{-2ex}=& \hspace{-2ex}Y_{u_i,u_j}^{(t)}-Y_{u_i,u_j}^{(t-1)}\notag\\
    &\hspace{-2ex}=& \hspace{-2ex}\left(\sum_{\forall v \in L(G), v \neq v_a}\left(w_{\angle(u_i, v, u_j)}\right)^2 + (w_{\angle(u_i, v_a, u_j)} + \lambda_{\angle(u_i, v_a, u_j)})^2 \right) \notag\\
    && -\sum_{\forall v \in L(G)}\left(w_{\angle(u_i, v, u_j)}\right)^2 \notag\\
    &\hspace{-2ex}=& \hspace{-2ex}2\lambda_{\angle(u_i, v_a, u_j)} w_{\angle(u_i, v_a, u_j)} + \lambda_{\angle(u_i, v_a, u_j)}^2, \notag
\end{eqnarray}
where $Y_{u_i,u_j}^{(t)}$ and $Y_{u_i,u_j}^{(t-1)}$ are the wedge-score-related data $Y_{u_i,u_j}$ at timestamps $t$ and $(t-1)$, respectively.

As given in Eq.~(\ref{eq:summed_wedge_score_increment}) and Eq.~(\ref{eq:summed_square_of_wedge_score_increment}), for any user pair $(u_i, u_j)$, the increase/decrease of $X_{u_i,u_j}$ is equal to $\lambda_{\angle(u_i, v_a, u_j)}$, and the increment/decrement of $Y_{u_i,u_j}$ can be computed via $w_{\angle(u_i, v_a, u_j)}$ and $\lambda_{\angle(u_i, v_a, u_j)}$.

\noindent {\bf Time Complexity Analysis:} 
For an updated edge $e_{u_i, v_a}$, we update the support upper bound of edges connected to $v_a$ and $v_b$, where $v_b \in N(u_i)\cap N(u_j)-\{v_a\}, u_j \in N(v_a)$. Thus, the worst-case time cost is $O(|N(v_a)|\cdot|N(u_i)|)$. 
Upon the update of edge $e_{u_i, v_a}$, for any user pair $(u_i, u_j)$, the time complexity of updating $X_{u_i,u_j}$ and $Y_{u_i,u_j}$ is given by $O(1)$. The worst-case time complexity is given by $O(|N(v_a)|)$, where $N(v_a)$ is a set of user neighbors $u_j$ of the item vertex $v_a$.
Therefore, the total time complexity of the graph increment maintenance is given by $O(|N(v_a)|\cdot|N(u_i)|)$.

\subsection{The Synopsis Construction}
\label{subsec:synopsis_construction}

We next present a hierarchical tree synopsis, $Syn$, on auxiliary data for dynamic graph $G_t$ (i.e., initial graph $G_0$ or its subsequent version at timestamp $t$) to accelerate online CD-SBN query processing.

\noindent{\bf The Data Structure of Synopsis, $Syn$:}
We design a hierarchical synopsis, $Syn$, over the initial bipartite graph $G_0$, where each synopsis node, $M$, has multiple entries $M_i$, and each of them corresponds to a partition of $G_0$. In detail, the hierarchical synopsis consists of two types of nodes: leaf and non-leaf nodes.


\underline{\it Leaf Nodes:} Each leaf node $M$ contains entries, each in the form of a user vertex $u \in U(G_t)$, associated with auxiliary data, $u.agg$. For each possible radius $r\in [1, r_{max}]$, we calculate and store aggregates in the form $(u.BV_r, u.ub\_sup_r, u.ub\_score_r)$ as follows:
\begin{itemize}
    \item a keyword bit vector, $u.BV_r$, of user vertex $u$, which is a bit-OR of bit vectors $v_i.BV$ for all item vertices $v_i$ in $(2r)$-hop subgraph of $u$ (i.e., $u.BV_r = \bigvee_{\forall v_i \in hop(u, 2r))} v_i.BV$);
    \item a support upper bound, $u.ub\_sup_r$ of support upper bounds for edges in the $(2r)$-hop subgraph of $u$ (i.e.,$u.ub\_sup_r = max_{\forall e \in hop(u, 2r)}{ub\_sup(e)}$), and;
    \item a user relationship score upper bound, $u.ub\_score_r$, of scores, $score_{u_i, u_j}(G_t)$, for all user pairs $(u_i, u_j)$ within the $(2r)$-hop subgraph of user $u$ (i.e., $u.ub\_score_r =  max_{\forall u_i, u_j \in hop(u, 2r)} \\score_{u_i, u_j}(G_t)$).
\end{itemize}

\underline{\it Non-Leaf Nodes:} Each non-leaf node $M$ in the synopsis $Syn$ contains entries $M_i$, each associated with the following aggregated data $M_i.agg$ (w.r.t. each possible radius $r \in [1, r_{max}]$).
\begin{itemize}
    \item an aggregated keyword bit vector $M_i.BV_r = \bigvee_{\forall u\in M_i} u.BV_r$;
    \item the maximum support upper bound $M_i.ub\_sup_r = \\ max_{\forall u \in M_i}{u.ub\_sup_r}$;
    \item the maximum score upper bound $M_i.ub\_score_r = \\ max_{\forall u \in M_i}{u.ub\_score_r}$, and;
    \item a pointer, $M_i.ptr$, pointing to a child node.
\end{itemize}

\noindent{\bf Bottom-up Construction of the Synopsis $\mathbf{Syn}$:}
We construct the hierarchical synopsis, $Syn$, as follows. For the initial bipartite graph $G_0$ (i.e., timestamp $t=0$), we start from each vertex $u_i \in U(G_0)$ and perform a breadth-first search (BFS) to extract a $(2r)$-hop subgraph, $hop(u_i, 2r)$, centered at $u_i$ with radius $2r$ (for $r\in [1, r_{max}]$). Next, for each subgraph $hop(u_i, 2r)$, we compute its auxiliary data, $u_i.agg$. We sort the vertices in $G_0$ according to their summed score upper bounds for all possible radii (i.e., $\sum_{\forall r\in [1, r_{max}]}{u.ub\_score_r}$), and divide them into multiple partitions of the same size $\gamma$ (i.e., vertices with similar score upper bounds are grouped together), which are leaf nodes of the synopsis $Syn$ (as discussed in Section~\ref{subsec:synopsis_construction}). Then, in a bottom-up manner, we recursively obtain non-leaf nodes $M$ of the synopsis $Syn$, by grouping leaf or non-leaf nodes $M_i$ (of the same size) on the lower level. The score upper bound associated with node $M$ is computed by aggregating the score upper bounds of its child nodes (i.e., $\sum_{\forall r\in [1, r_{max}]}$ $M_i.ub\_score_r$). The recursive synopsis construction terminates until one root node, $root(Syn)$, is obtained. After building the tree structure of the synopsis, we will recursively calculate aggregates, $M_i.agg$, of entries $M_i$ in leaf/non-leaf nodes of synopsis $Syn$. In addition, we keep an inverted list, $inv\_list$, for user vertices $u_i \in U(G_t)$, where each user vertex $u_i$ is associated with an ordered list of node IDs from root to a leaf node in synopsis $Syn$ in which $u_i$ resides. This inverted list, $inv\_list$, can be used for accessing a particular user vertex for the synopsis maintenance (as discussed in Section~\ref{subsec:incremental_maintenance_synopsis}).



\noindent{\bf Complexity Analysis:}
For a tree index $\mathcal{I}$, we denote the fanout of each non-leaf node $N$ as $\gamma$ and the average number of vertex degrees as $avg\_deg$. To compute the auxiliary data that we mentioned above, it takes $O(2r_{max})$ to aggregate the item keyword bit vector ($v_i.BV_r$), support upper bound ($ub\_sup_r$), and user relationship score upper bound $ub\_score_r$ from child nodes.
Since the number of leaf nodes is equal to the number of user vertices $|U(G)|$ divided by $\gamma$, the depth of the tree index $\mathcal{I}$ is $\lceil \log_{\gamma}{|U(G)|} \rceil$. The time complexity of constructing a $\gamma$-ary tree with a depth of $dep$ is $O(\gamma^{dep}-1)/gamma-1$. Therefore, the total time complexity of our tree index construction is given by $O(2r_{max}\cdot(\gamma^{\lceil \log_{\gamma}{|U(G)|} \rceil}-1)/(\gamma-1))$.

For the space complexity, each leaf node contains a keyword bit vector of size $B$, a support upper bound, and a user relationship score upper bound. Therefore, the space complexity of all leaf nodes is $O(2r_{max}\cdot B \cdot |V(G)|/\gamma)$. On the other hand, for each non-leaf node, it needs a bit vector sized of $B$, a support upper bound, a user score upper bound, and a list of pointers with size $\gamma$ to its child nodes. Since the depth of the tree index is $\lceil \log_{\gamma}{|V(G)|} \rceil$,  the space complexity of all non-leaf nodes is $O(2r_{max} \cdot (B+\gamma) \cdot (\gamma^{\lceil \log_{\gamma}{|V(G)|} \rceil-1}-1)/(\gamma-1))$

Overall, the space cost of our tree index is $O(2r_{max} \cdot ((B+\gamma) \cdot (\gamma^{\lceil \log_{\gamma}{|V(G)|} \rceil-1}-1/(\gamma-1)) + B \cdot |V(G)|/\gamma)))$.

\subsection{Incremental Maintenance of the Synopsis}
\label{subsec:incremental_maintenance_synopsis}
In this subsection, we discuss how to identify the affected vertices and nodes in the synopsis $Syn$ and update their aggregates in $Syn$. 

\noindent {\bf Identification of the Affected Vertices and Synopsis Nodes:} 
When an item from $S$ arrives or expires in $W_t$ (i.e., an edge $e_{u_i, v_a}$ changes its weight $w_{u_i, v_a}$ or is inserted/removed into/from graph $G_{t-1}$), the auxiliary data, $u.agg$, of some user vertices need to be updated. Therefore, our first goal is to identify those potentially affected vertices (with aggregate updates), which are user vertices in the $(2r)$-hop subgraph of $u_i$ for all possible radii $r\in[1,r_{max}]$.

Specifically, if the weight of edge $e_{u_i, v_a}$ is updated, we will consider those users $u \in hop(u_i, 2r)$ within $(2r)$ hops away from $u_i$ as the affected vertices, which need to update keyword bit vectors $BV_r$. Moreover, since the support of each edge $e$ in $\Join(u_i, u_j, v_a, v_b)$ may increase/decrease, user vertices, $u$, within $(2r)$-hop from $u_i$ may be affected (as $(2r)$-hop subgraph of $u$ may contain those edges $e$). Moreover, the user relationship score between $u_i$ and each vertex $u_j \in N(v_a)$ may change. Thus, we need to update the scores of user vertices $u_j$ whose $(2r)$-hop subgraphs contain both $u_i$ and $u_j$.

Finally, in synopsis $Syn$, we identify those (non-)leaf nodes that contain the affected user vertices and update their aggregated data.

\noindent {\bf Maintenance of Leaf and Non-leaf Nodes:}
After identifying the affected vertices/nodes, we update their aggregates for all possible radii $r\in[1,r_{max}]$. Assume that an edge $e_{u_i, v_a}$ is inserted/expired, and $u$ is an affected vertex.
For the edge insertion, we merge $v_a.BV$ into $u.BV_r$, update $u.ub\_sup_r$ and $u.ub\_score_r$ for each affected vertex $u$, and update aggregates of the affected synopsis nodes $M$ containing $u$ (via $inv\_list$).
For the edge expiration, we re-aggregate the data, $u.BV_r$, $u.ub\_sup_r$, and $u.ub\_score_r$, from $(2r)$-hop of each affected vertex $u$, and we re-aggregate the data in the affected nodes containing $u$ (if the re-aggregated data on the lower-level affected nodes have changed).

\noindent {\bf Maintenance of the Newly Added User Vertices:} 
When an edge is inserted with a new user vertex $u_{new}$, we first compute the aggregates, $u_{new}.agg$ for this new user $u_{new}$ (within $(2r)$ hops from $u_{new}$) and add $u_{new}$ to a leaf node $M_i$ in synopsis $Syn$ (containing a user vertex with the closest score upper bound). Next, we update the aggregate data of non-leaf nodes containing leaf node $M_i$.

\noindent {\bf Time Complexity Analysis:}
Assume that an edge $e_{u_i, v_a}$ is inserted/expired. To identify affected vertices, it takes $O(|U(hop(u_i, 2r)|)$ to obtain vertices that may need to update $BV_r$, $u.ub\_sup_r$, and $u.ub\_score_r$. After that, we check $inv\_list$ to retrieve the affected nodes containing the affected vertices with the time complexity of $O(1)$.
When edge $e_{u_i,v_a}$ is inserted, it takes $O(1)$ to merge $v_a.BV$ into $u.BV_r$, maintain support upper bound $u.ub\_sup_r$ and score upper bound $u.ub\_score_r$ for each affected vertex $u$. Then, we update the aggregates in leaf and non-leaf nodes containing $u$ with a cost of $O(1)$. For the edge expiration, it takes $O(|U(hop(u,$ $2r))|)$ time complexity to aggregate the bit vector of user keywords, calculate support upper bound $u.ub\_sup_r$, and update score upper bound $u.ub\_score_r$ according to the $(2r)$-hop subgraph of each affected vertex $u$. To maintain non-leaf nodes, we recompute the aggregates for each non-leaf node, which takes $O(\gamma)$, where $\gamma$ is the average size of synopsis nodes.
In the worst case, we can find $|U(hop(u_i, 2r))|$ affected vertices and $h$ synopsis nodes for each vertex, where $h$ is the height of the tree synopsis $Syn$. In summary, it takes $O(|U(hop(u_i, 2r))| + h)$ for each edge insertion and $O(|U(hop(u_i, 2r))|\cdot|U(hop(u, 2r))| + h \cdot \gamma)$ for each edge expiration.

\section{CD-SBN query Processing}
\label{sec:query_processing}


\subsection{Pruning via Synopsis}
\label{subsec:pruning_via_synopses}


\noindent{\bf Keyword Pruning for Synopsis Entries:}
We first discuss how to use the aggregated keyword bit vector, $M_i.BV_r$, to rule out a synopsis entry $M_i$, if no query keyword in $Q$ exists in the vertices under $M_i$. \nop{{\color{blue} Please refer to complete proofs in our technical report on Arxiv~\cite{zhang2024EffectiveCommunityDetection}.}}

\begin{lemma}
\label{lemma:synopsis_keyword_pruning}
    {\bf (Synopsis Keyword Pruning)} Given a synopsis entry $M_i$ and a set, $Q$, of query keywords, entry $M_i$ can be safely pruned, if it holds $M_i.BV_r \bigwedge Q.BV = \boldsymbol{0}$, where $Q.BV$ is a bit vector hashed from query keywords in $Q$.
\end{lemma}

\begin{proof}
\label{proof:synopsis_keyword_pruning}
Since $M_i.BV$ is an aggregated keyword bit vector computed by $\bigvee_{\forall v\in M_i} v.BV$, the $f(w)$-th bit position of $M_i.BV_r$ with the value $0$ means that the keyword $w$ is not contained in the keyword set, $v.K$, of any item vertex $v$, where $v \in N(u)$ and the $f(\cdot)$ is a hash function. Therefore, if $M_i.BV_r \wedge Q.BV = \boldsymbol{0}$ holds, no query keywords in $Q$ belong to $M_i$, which violates the keyword constraint in Definition \ref{def:cd_sbn_problem} (i.e., $v.K \cap Q \ne \emptyset$). In this case, the synopsis entry $M_i$ cannot contain any vertices in the CD-SBN query answer, and thus it can be safely pruned, completing the proof. 
\end{proof}

\nop{
\begin{proof}
Please refer to the proof in our technical report~\cite{zhang2024EffectiveCommunityDetection}.
\end{proof}
}

\noindent{\bf Support Pruning for Synopsis Entries:}
Next, we utilize the maximum support upper bound, $M_i.ub\_sup_r$, and the query support threshold $k$ to filter out a synopsis entry $M_i$. 


\begin{lemma}
\label{lemma:synopsis_support_pruning}
    {\bf (Synopsis Support Pruning)} Given a synopsis entry $M_i$ and a support threshold $k$, entry $M_i$ can be safely pruned, if it holds that $M_i.ub\_sup_r < k$, where $M_i.ub\_sup_r$ is the maximum edge support upper bound in all $(2r)$-hop subgraphs of vertices under $M_i$.
    \label{lemma:index_support_pruning}
\end{lemma}

\begin{proof}
\label{proof:synopsis_support_pruning}
    Since $M_i.ub\_sup_r$ is the maximum edge support upper bound for all vertices $u \in M_i$, it holds that $M_i.ub\_sup_r \geq u.ub\_sup_r$. Moreover, $u.ub\_sup_r$ is greater than the support of any edge in the $2r$-hop subgraph of $u$. Therefore, if $M_i.ub\_sup_r < k$ holds, the support of any edge connected to $u \in M_i$ is less than $k$ (for any user vertex $u \in M_i$), which indicates that $M_i$ can be safely pruned. 
\end{proof}

\nop{
\begin{proof}
Please refer to the proof in our technical report~\cite{zhang2024EffectiveCommunityDetection}.
\end{proof}
}

\noindent{\bf Score Upper Bound Pruning for Synopsis Entries:}
Finally, we exploit the maximum upper bound, $M_i.ub\_score_r$, of user relationship scores and the score threshold $\sigma$ to discard synopsis entries.


\begin{lemma}
\label{lemma:synopsis_score_upper_bound_pruning}
    {\bf (Synopsis Score Upper Bound Pruning)} Given a synopsis entry $M_i$ and a user relationship score threshold $\sigma$, entry $M_i$ can be safely pruned, if it holds that $M_i.ub\_score_r < \sigma$.
\end{lemma}

\begin{proof}
\label{proof:synopsis_score_upper_bound_pruning}
    Since $M_i.ub\_score_r$ is the maximum score upper bound, it holds that $M_i.ub\_score_r \geq u.ub\_score_r$, where $u.ub\_score_r$ is the user relationship score upper bound of any user vertex $u \in M_i$. Moreover, it holds that $u.ub\_score_r \geq score_{u, u_l}(G_t)$, where $score_{u, u_l}(G_t)$ is the user relationship score between $u$ and its $2$-hop reachable user vertex $u_l$. 
    By the inequality transition, we have $M_i.ub\_score_r \geq u.ub\_score_r\geq score_{u, u_l}(G_t)$. If $M_i.ub\_score_r < \sigma$ holds, then we have $score_{u, u_l}(G_t)<\sigma$ for all $u$ under $M_i$, which violates the score constraint. Therefore, in this case, we can safely prune entry $M_i$,  which completes the proof. 
\end{proof}

\nop{
\begin{proof}
Please refer to the proof in our technical report~\cite{zhang2024EffectiveCommunityDetection}.
\end{proof}
}

\subsection{Snapshot CD-SBN Processing Algorithm}
\label{subsec:snapshot_processing_algorithm}

Algorithm~\ref{algo:snapshot_cd_sbn_algorithm} details our snapshot CD-SBN processing algorithm to obtain initial CD-SBN community answers, which consists of three major parts: i) the initialization of data structures and variables for preparing the synopsis traversal (lines 1-4); ii) the traversal of the synopsis, $Syn$, to obtain a candidate set $P$ (lines 5-18); and iii) the refinement of a snapshot CD-SBN result set $R$ (line 19).

\setlength{\textfloatsep}{0pt}
\begin{algorithm}[t]
\caption{{\bf Snapshot CD-SBN Processing}}\small
\label{algo:snapshot_cd_sbn_algorithm}
\KwIn{
    \romannumeral1) a streaming bipartite network $G_t$;
    \romannumeral2) a tree synopsis $Syn$ over $G_t$, and;
    \romannumeral3) a snapshot CD-SBN query with a set, $Q$, of query keywords, a query support threshold $k$, a query maximum radius $r$, and a user relationship score threshold $\sigma$
}
\KwOut{
    a set, $R$, of CD-SBN communities
}

\tcp{\bf Initialization}
hash all query keywords in $Q$ into a bit vector $Q.BV$\;

initialize a maximum heap $\mathcal{H}$ in the form of $(M, key)$\;
insert $(root(Syn), +\infty)$ into heap $\mathcal{H}$\;

$P = \emptyset$; $R = \emptyset$\;

\tcp{\bf Synopsis Traversal}
\While{$\mathcal{H}$ is not empty}{
    $(M, key)$ = de-heap$(\mathcal{H})$\;
    \If(\tcp*[h]{Lemma~\ref{lemma:synopsis_score_upper_bound_pruning}}){$key < \sigma$}{terminate the loop\;}

    \eIf{$M$ is a leaf node}{
        \For{each vertex $u_i \in M$}{
            \If{$2r$-hop subgraph $hop(u_i, 2r)$ cannot be pruned by Lemma~\ref{lemma:keyword_pruning}, \ref{lemma:support_pruning},
            \ref{lemma:layer_size_pruning}, or \ref{lemma:score_upper_bound_pruning}}{
                obtain a maximal $(k,r,\sigma)$-bitruss $g \subseteq hop(u_i, 2r)$ satisfying the keyword constraint\;
                \If{$g$ exists}{
                    add $g$ to $P$\;
                }
            }
        }
    
    }(\tcp*[h]{$M$ is a non-leaf node}){
        \For{each entry $M_i \in M$}{
            \If{$M_i$ cannot be pruned by Lemma~\ref{lemma:synopsis_keyword_pruning}, \ref{lemma:synopsis_support_pruning}, or \ref{lemma:synopsis_score_upper_bound_pruning}}{
                insert entry $(M_i, M_i.ub\_score_r)$ into heap $\mathcal{H}$\;
            }
        }
    }

}

\tcp{\bf Refinement}
refine candidate communities in $P$ (removing the redundancy) and obtain the CD-SBN result set $R$\;

\Return $R$\;
\end{algorithm}
\setlength{\textfloatsep}{0pt}

\noindent{\bf Initialization:} Given a query keyword set $Q$, the algorithm first hashes all query keywords in $Q$ into a query keyword bit vector $Q.BV$ (line 1). Then, we initialize a \textit{maximum heap} $\mathcal{H}$ (for synopsis traversal), accepting heap entries in the form of $(M, key)$, where $M$ is a synopsis entry and $key$ is the maximum score upper bound, $M.ub\_score_r$, of all user vertices in $M$ (as mentioned in Section~\ref{subsec:synopsis_construction}) (line 2). Intuitively, a heap entry with a large key (i.e., a large maximum score upper bound) is more likely to contain communities with high user relationship scores. Next, we insert the root of the tree synopsis $Syn$, in the form $(root(Syn), +\infty)$, into the heap $\mathcal{H}$ (line 3). Moreover, we use an initial empty set $P$ to keep the candidate subgraphs and an empty set $R$ to store snapshot CD-SBN query answers (line 4). 




\noindent{\bf Synopsis Traversal:}
Next, we utilize the maximum heap $\mathcal{H}$ to traverse the synopsis (lines 5-18). Each time we pop out a heap entry $(M, key)$ from heap $\mathcal{H}$ with the highest key, $key$ (i.e., the maximum score upper bound) (lines 5-6). If $key$ is less than the score threshold, $\sigma$, it indicates that all the remaining heap entries in $\mathcal{H}$ have a score less than $\sigma$, and they cannot be contained in our CD-SBN query results. Thus, we can safely terminate the synopsis traversal early (lines 7-8).


When $M$ is a leaf node, we consider each user vertex $u_i \in M$ and check whether a candidate community centered on $u_i$ exists (lines 9-14). For each user $u_i\in M$, if the $2r$-hop subgraph $hop(u_i, 2r)$ cannot be pruned by our pruning strategies (i.e., \textit{keyword}, \textit{support}, and \textit{score upper bound pruning}), then we will obtain a maximal $(k,r,\sigma)$-bitruss $g$ (satisfying the keyword constraint) within $hop(u_i, 2r)$ (lines 10-12). If such a subgraph $g$ exists, we will add $g$ to the candidate community set $P$ (lines 13-14). 


On the other hand, when $M$ is a non-leaf node, we will check each entry $M_i$ of node $M$ (lines 15-16). If synopsis entries $M_i$ cannot be pruned by synopsis-level pruning strategies, then we will insert $(M_i, M_i.ub\_score_r)$ into the heap $\mathcal{H}$ for next checking (lines 17-18).


\noindent{\bf Refinement:} 
Finally, we remove duplicate subgraphs from $P$ (due to overlap of $2r$-hop subgraphs), refine candidate communities in $P$ to obtain CD-SBN query answers in $R$, and return $R$ (lines 19-20).


\noindent{\bf Complexity Analysis:}
Let $PP^{(j)}$ be the pruning power (i.e., the percentage of node entries that can be pruned) at the $j$-th level of synopsis $Syn$, where $0 \leq j \leq h$ (here $h$ is the height of the tree synopsis $Syn$). Denote $f$ as the average fanout of the synopsis nodes in $\mathcal{I}$. For the index traversal, the number of nodes that need to be accessed is given by $\sum_{j=1}^{h}f^{h-j+1}\cdot(1 - {PP}^{(j)})$.
Next, to obtain $g$ satisfying the constraints, we need to prune the item vertices without any query keyword. Therefore, it is necessary to recompute the support of the relevant edges and the user relationship score of any pair of adjacent user vertices. Let $\overline{n}$ be the number of pruned item vertices and $deg_{avg}$ be the average degree of the item vertices. The time cost of re-computation is $\overline{n} \cdot deg_{avg}^2$ (given in Section~\ref{subsec:incremental_maintenance_graph} and Section~\ref{subsec:incremental_maintenance_synopsis}).
Then, we can obtain the maximal $(k,r,\sigma)$-bitruss in $O(deg_{avg}^2)$ with computed edge support, which is demonstrated in~\cite{sariyuce2018PeelingBipartiteNetworks}. Thus, it takes $O(f^{h+1}\cdot(1-{PP}^{(0)}) \cdot \overline{n} \cdot deg_{avg}^2)$  to refine candidate seed communities, where ${PP}^{(0)}$ is the pruning power over $r$-hop subgraphs in leaf nodes.
Therefore, the total time complexity of Algorithm~\ref{algo:snapshot_cd_sbn_algorithm} is given by $O\left(\sum_{j=1}^{h}f^{h-j+1}\cdot(1 - {PP}^{(j)})+f^{h+1}\cdot(1-{PP}^{(0)}) \cdot \overline{n} \cdot deg_{avg}^2\right)$.

\begin{algorithm}[t]
\caption{{\bf Continuous CD-SBN Processing}}\small
\label{algo:continuous_cd_sbn_algorithm}
\KwIn{
    \romannumeral1) a streaming bipartite network $G_{t-1}$ with an update stream $S$;
    \romannumeral2) a set, $R_{t-1}$, of community answers at timestamp $(t-1)$, and;
    \romannumeral3) a continuous CD-SBN query with a set, $Q$, of query keywords, a support threshold $k$, a maximum radius $r$, and a threshold, $\sigma$, of user relationship score
}
\KwOut{
    a set, $R_t$, of CD-SBN community answers at timestamp $t$
}
$R_t = R_{t-1}$\;

\tcp{\bf Processing the Expired Item $p_{t-s} = (e_{u_i, v_a}, t-s)\in S$}

\For{each community answer $g \in R_t$}{
    \If{$e_{u_i, v_a} \in E(g)$}{
    update the CD-SBN community $g$ with $g'$, upon the weight update (or deletion) of edge $e_{u_i, v_a}$\;
    \If{$g'$ does not exist}{
        remove $g$ from $R_t$\;
    }
    }
}
\tcp{\bf Processing the Insertion of a New Item $p_t = (e_{u_i', v_a'}, t)\in S$ }
    $\Delta R = \emptyset$\;
    \For{each user vertex $u_j \in hop(u_i', 2r)$}{
        obtain a CD-SBN candidate community $g$ from $hop(u_j, 2r)$ (via pruning) and add it to $\Delta R$\;
    }
    remove duplicate communities $g$ from $\Delta R$ (overlapping with $R_t$) and refine candidate communities $g$ in $\Delta R$\;
    $R_t = R_t \cup \Delta R$\;
    \Return $R_t$\;
\end{algorithm}

\subsection{Continuous CD-SBN Processing Algorithm}
\label{subsec:continuous_processing_algorithm}

Algorithm~\ref{algo:continuous_cd_sbn_algorithm} details our continuous CD-SBN query processing approach. First, we initialize the CD-SBN result set $R_t$ at timestamp $t$ with that $R_{t-1}$ at timestamp $(t-1)$ (line 1). Next, we check the constraints of the existing communities to refine the previous result set (lines 2-6). Then, we obtain the potential communities from the $2r$-hop of user vertices near the insertion edge and remove duplicates (lines 7-10). Finally, we merge the candidate communities into the result set and return it at timestamp $t$ (lines 11-12).

\noindent{\bf Edge Expiration Processing:}
To deal with the expiration of an item $p_{t-s}$ ($ = (e_{u_i, v_a}, t-s)$) in the sliding window $W_{t-1}$, the weight $w_{u_i, v_a}$ of the edge $e_{u_i, v_a}$ decreases (which might turn to the edge deletion). In this case, we will identify those community answers $g$ in the answer set $R_t$ that contain edge $e_{u_i, v_a}$ (lines 2-3), and update the community $g$ with $g'$ (upon weight/edge changes; line 4). If such a community $g'$ does not exist (violating the CD-SBN predicates), we simply remove community answer $g$ from $R_t$ (lines 5-6).

\noindent{\bf Edge Insertion Processing:}
Next, for the insertion of a new item $p_t$ ($= (e_{u_i', v_a'}, t)$), we aim to update $R_t$ with new CD-SBN community answers (containing the new/updated edge $e_{u_i', v_a'}$) in a set $\Delta R$ (lines 7-11). In particular, for each user vertex $u_j$ within the $2r$-hop away from $u_i'$, we obtain a CD-SBN candidate community $g$ via the pruning strategies mentioned in Section \ref{sec:pruning_strategies} and add it to $\Delta R$ (if $g$ exists) (lines 8-9). Then, we will remove redundancy among community candidates in $\Delta R$ (i.e., those already in $R_t$) and refine the remaining candidate communities in $\Delta R$ (line 10). Finally, we add CD-SBN answers in $\Delta R$ to $R_t$, and return final CD-SBN answers in $R_t$ (lines 11-12).

\noindent{\bf Complexity Analysis:}
Let $deg_{avg}$ be the average degree of the item vertices. For edge insertion and expiration, the time cost of user relationship score re-computation is $O(deg_{avg})$ and of edge support update is $O(deg_{avg}^2)$. Since obtaining the maximal $(k,r,\sigma)$-bitruss with computed edge supports and user relationship scores~\cite{sariyuce2018PeelingBipartiteNetworks} requires $O(deg_{avg}^2)$, it takes $O(|R_{t-1}|deg_{avg}^2)$ to maintain the result set at $t-1$ and $O(deg_{avg}^3)$ to detect potential communities.
Therefore, the total time complexity of Algorithm~\ref{algo:continuous_cd_sbn_algorithm} at each timestamp $t$ is given by $O\left(|R_{t-1}|deg_{avg}^2 + deg_{avg}^3\right)$.

\begin{table}[t]
\begin{center}
\caption{Statistics of graph data sets.}
\vspace{-2ex}
\label{tab:datasets}
\footnotesize
\begin{tabular}{c||c|c|c|c|c}
\toprule
\textbf{Data Set}& \textbf{Type}  & $|U|$ & $|V|$ & $|E|$ & $|\Sigma|$ \\
\midrule
    AmazonW (AM) & Rating & 26,112 & 800 & 111,265 & N/A \\
    BibSonomy (BS) & Publication & 5,795 & 767,448 & 2,555,080 & 204,674 \\
    CiaoDVD (CM) & Rating & 17,616 & 16,121 & 295,958 & N/A \\
    CiteULike (CU) & Publication& 22,716 & 731,770 & 2,411,819 & 153,278 \\
    Movielens (ML) & Movie & 4,010 & 7,602 & 95,580 & 16,529 \\
    Escorts (SX) & Rating & 10,106 & 6,624 & 50,632 & N/A \\
    TripAdvisor (TA) & Rating & 145,317 & 1,760 & 703,171 & N/A \\
    UCForum (UF) & Interaction & 899 & 522 & 33720 & N/A \\
    ViSualizeUs (VU) & Picture & 17,110 & 495,402 & 2,298,816 & 82,036 \\
\hline
    dBkU, dBkL, dBkP & Synthetic & 10K$\sim$100K & 10K$\sim$100K & 152,175 & N/A \\
    dPkU, dPkL, dPkP & Synthetic & 10K$\sim$100K & 10K$\sim$100K & 152,922 & N/A \\
\bottomrule
\end{tabular}
\end{center}
\end{table}

\section{Experimental Evaluation}
\label{sec:experiments}

\subsection{Experimental Settings}
\label{subsec:experimental_settings}
We conduct experiments to evaluate the performance of our snapshot and continuous CD-SBN query processing algorithms on various bipartite graphs. Our source code is available on GitHub\footnote{\url{https://github.com/L1ANLab/CD-SBN}}.

\noindent {\bf Bipartite Graph Data Sets:}
We evaluate our proposed CD-SBN algorithms on nine real-world and six synthetic graphs. As depicted in Table~\ref{tab:datasets}, we provide statistics of real-world bipartite graphs from the KONECT\footnote{\url{http://konect.cc/}} project (i.e., Koblenz Network Collection), where $|\Sigma|$ refers to the domain size of keywords in item vertices of the original graphs (``N/A'' for $|\Sigma|$ means that graph data have no keywords). 

To generate synthetic bipartite graphs, we first randomly produce degrees, $deg(u)$, of user vertices $u$ following \textit{PowerLaw} or \textit{Beta} distribution. Then, for each user vertex $u$, we connect it to $deg(u)$ random item vertices $v$.  Next, to simulate the real-world frequency of user-item interaction (i.e., edge weights $w_{u,v}$), we assign each edge with weight ranging from $\left[min\_w, max\_w\right]$ following the \textit{Gaussian} distribution, and associate them with edges, where the mean and standard deviation of the \textit{Gaussian} distribution are $(1.5, 0.25)$, $(2, 0.5)$, and $(2.5, 0.75)$ for $\left[1,2\right]$, $\left[1,3\right]$, and $\left[1,4\right]$, respectively.
Finally, for each item vertex $v$, we generate its keyword set $v.K$, which contains integers following \textit{Log-Normal}, \textit{Pareto}, or \textit{Uniform} distribution. 
For real-world graphs without keywords (e.g., AM, CM, SX, UF, and TA), we produce keyword sets of their item vertices following the \textit{Log-Normal} distribution.

For CD-SBN predicates, we will randomly select $|Q|$ keywords from a keyword domain $\Sigma$, following the keyword distribution in the data graph $G$, and form a set, $Q$, of query keywords. {\color{blue} } Other parameter settings are provided in Table \ref{tab:parameters}.

\begin{table}[t!]
\begin{center}
\caption{Parameter settings.}
\vspace{-3ex}
\label{tab:parameters}
\footnotesize
\begin{tabular}{l||p{30ex}}
\hline
\toprule
\textbf{Parameters}&\textbf{Values} \\
\midrule
    sliding window size $s$ & 200, 300 \textbf{500}, 800, 1000\\
    support, $k$, of bitruss structure & 3, \textbf{4}, 5\\  
    radius $r$ & 1, \textbf{2}, 3\\
    relationship score threshold $\sigma$ &1, 2, \textbf{3}, 4, 5\\
    size, $|Q|$, of query keyword set $Q$  & 2, 3, \textbf{5}, 8, 10 \\
    keyword domain size $|\Sigma|$ & 100, 200, \textbf{500}, 800, 1000\\
    size, $|v_i.K|$, of keywords per item vertex & 1, 2, {\bf 3}, 4, 5 \\
    the number, $|U(G)|$, of user vertices & 10K, 15K, \textbf{25K}, 50K, 100K\\
    the number, $|L(G)|$, of item vertices & 10K, 15K, \textbf{25K}, 50K, 100K\\
    the edge weight range $\left[min\_w, max\_w\right]$ & \textbf{[1,2]}, [1,3], [1,4]\\
\bottomrule
\end{tabular}
\end{center}
\end{table}

\nop{
\underline{\it Analysis of Degree and Keyword Distribution:}
We performed a statistical analysis to fit the distribution of user vertex degrees in real-world graphs.
We first counted the degrees of all the user vertices of each real-world graph, then utilized the goodness-of-fit test to determine the probability of distribution by comparing the observed frequency to the expected frequency from the model (f-hat) and computing the residual sum of squares (RSS), and finally obtained the fitting probability ranking and its parameters for 16 common distributions.
Combining the best probability distribution rankings on all real-world graphs, we selected \textbf{PowerLaw} and \textbf{Beta} distributions as the sample distributions for user vertex degrees. We use the same method to analyze the keyword distributions in real-world graphs with original keywords (\textbf{BS}, \textbf{CU}, \textbf{ML}, and \textbf{VU}) and chose \textbf{Log-Normal}, \textbf{Pareto}, and \textbf{Uniform} distributions as the sample distributions for keywords set associated with user vertices.

For each query, we randomly select $|Q|$ keywords from the keyword domain $\Sigma$ following the distribution of the keywords in the data graph and form a set of query keywords $Q$.
}

\noindent {\bf Competitors:}
Since no prior work has studied the CD-SBN problem under the same semantics as $(k,r,\sigma)$-bitruss communities with specific keywords, we compare our CD-SBN algorithm with a baseline method, \textit{Bitrussness-Based Decomposition (BBD)}.
Specifically, \textit{BBD} first calculates the trussness of each edge, which is the maximum $k$ of all the $k$-bitruss containing the edge, and builds a synopsis based on the trussness. Then, upon edge updates at each timestamp, \textit{BBD} filters out user vertices connected to edges with trussness less than $k$ and computes the $(k,r,\sigma)$-bitruss community centered at each remaining vertex (containing query keywords).

\noindent {\bf Measures:}
We evaluate our CD-SBN performance in terms of the \textit{wall clock time}, which is the time cost to traverse the index and retrieve communities for snapshot CD-SBN, or that to maintain the CD-SBN results at each timestamp for the continuous CD-SBN problem. The wall clock time is the average time across 10 runs with a set of 10 different CD-SBN predicates (e.g., query keywords).


\setlength{\textfloatsep}{0pt}
\subfigcapskip=-0.2cm
\begin{figure*}[t!]
    \centering
    \subfigure[sliding window size $s$]{
        \includegraphics[height=2.7cm]{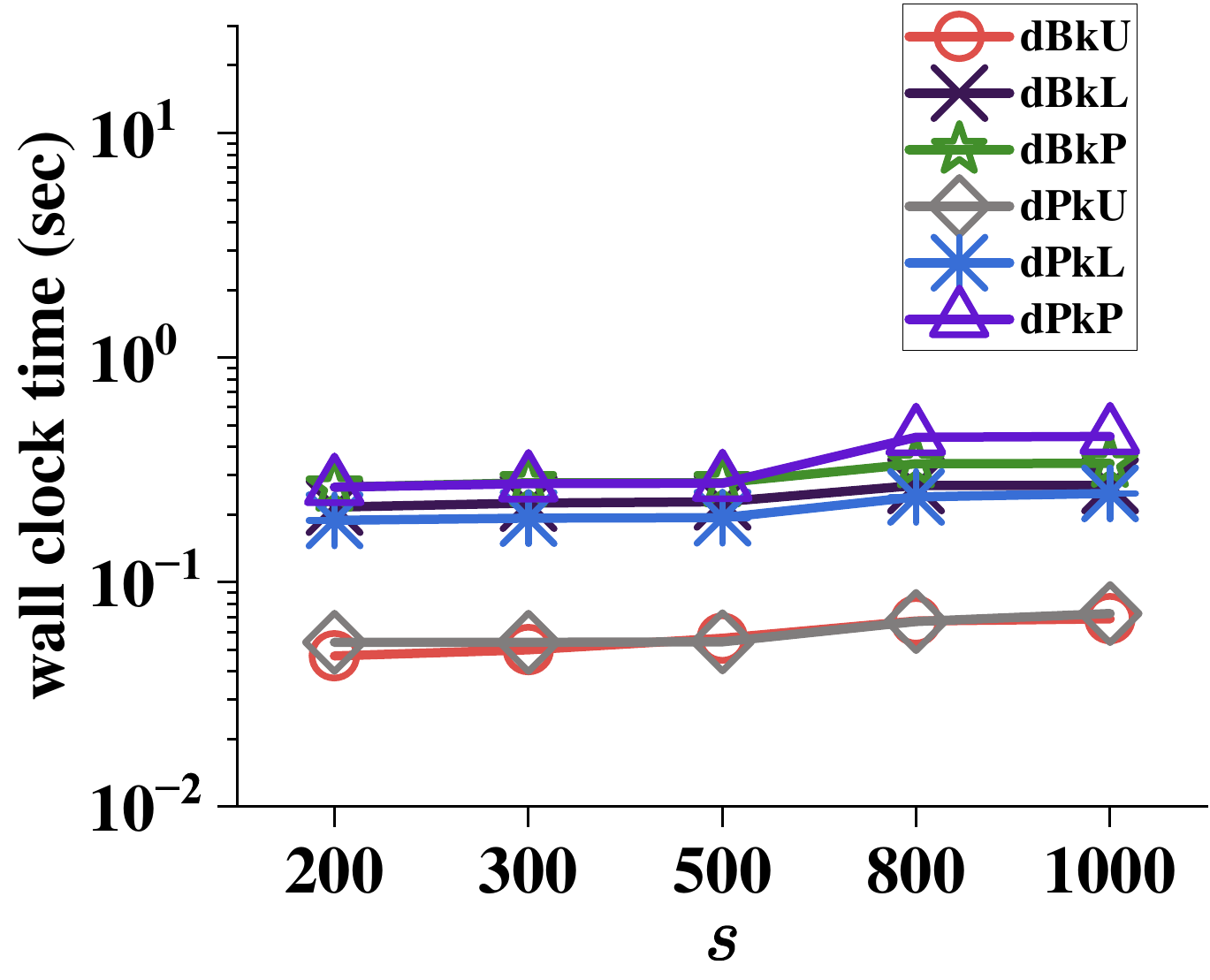}
        \label{subfig:effect_sliding_window_size}
    }\quad
    \hspace{-0.3cm}
    \subfigure[bitruss support threshold $k$]{
        \includegraphics[height=2.7cm]{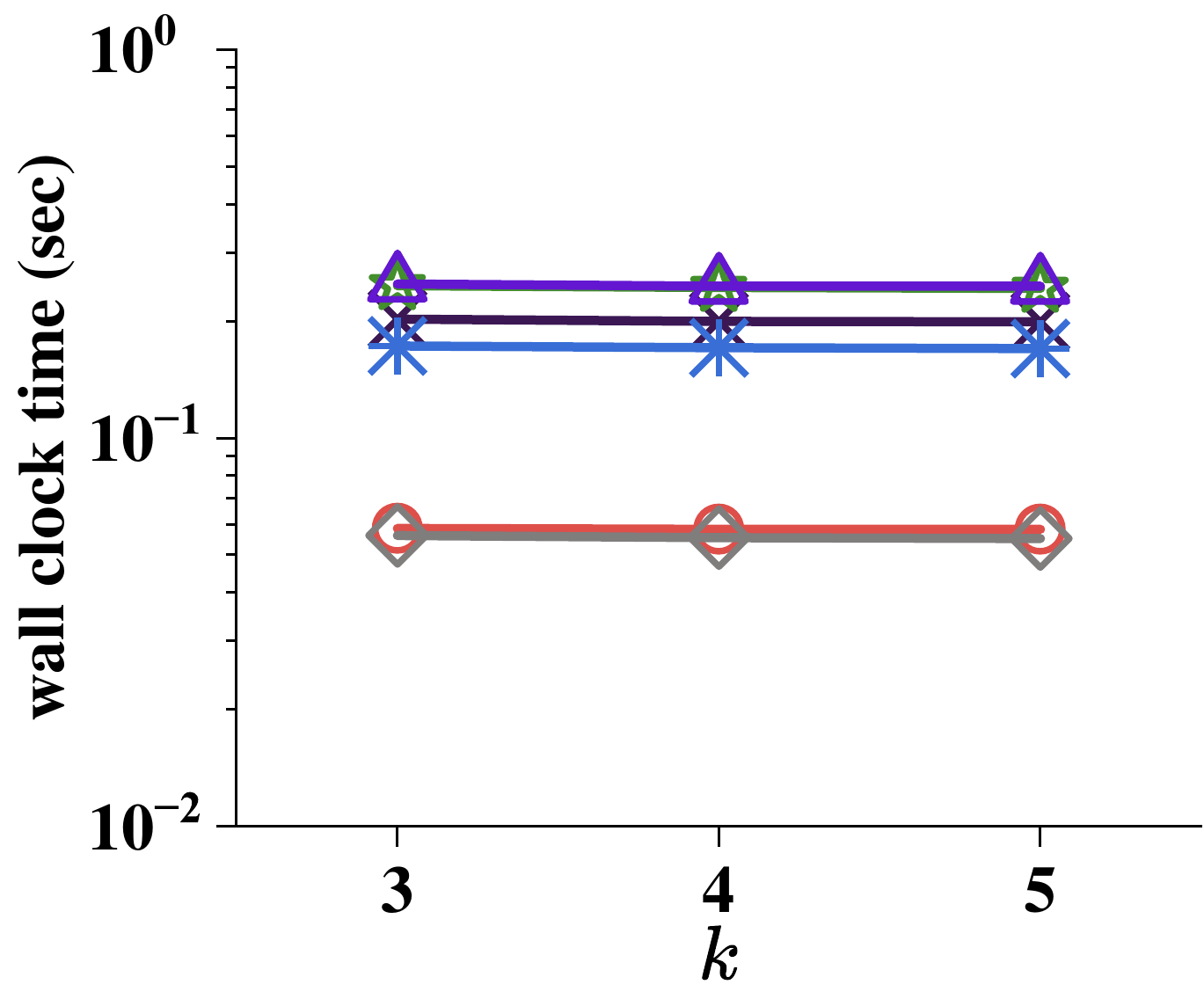}
        \label{subfig:effect_truss_support_parameter}
    }\quad
    \hspace{-0.3cm}
    \subfigure[radius $r$]{
        \includegraphics[height=2.7cm]{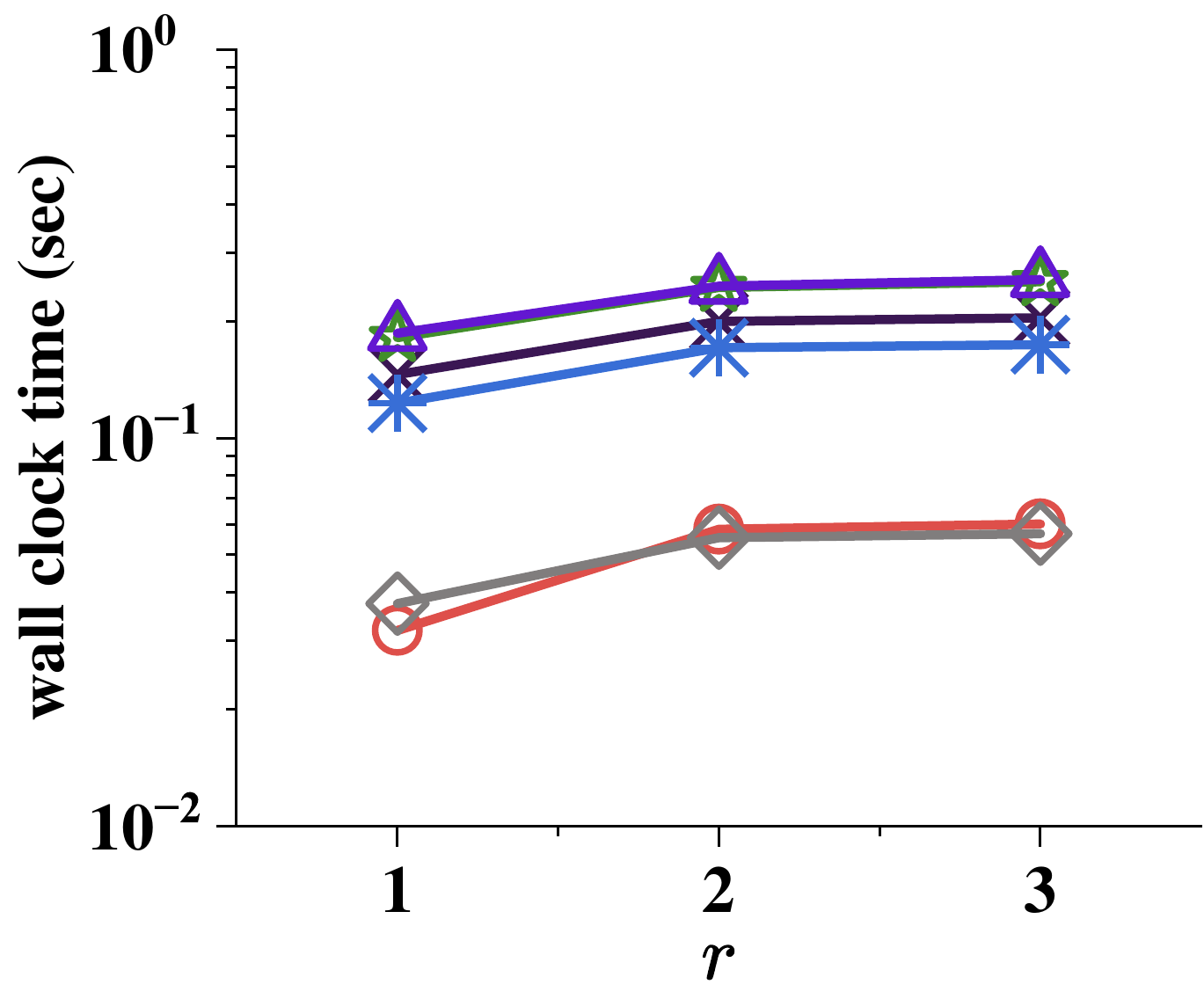}
        \label{subfig:effect_radius}
    }\quad
    \hspace{-0.3cm}
    \subfigure[relationship score threshold $\sigma$]{ 
        \includegraphics[height=2.7cm]{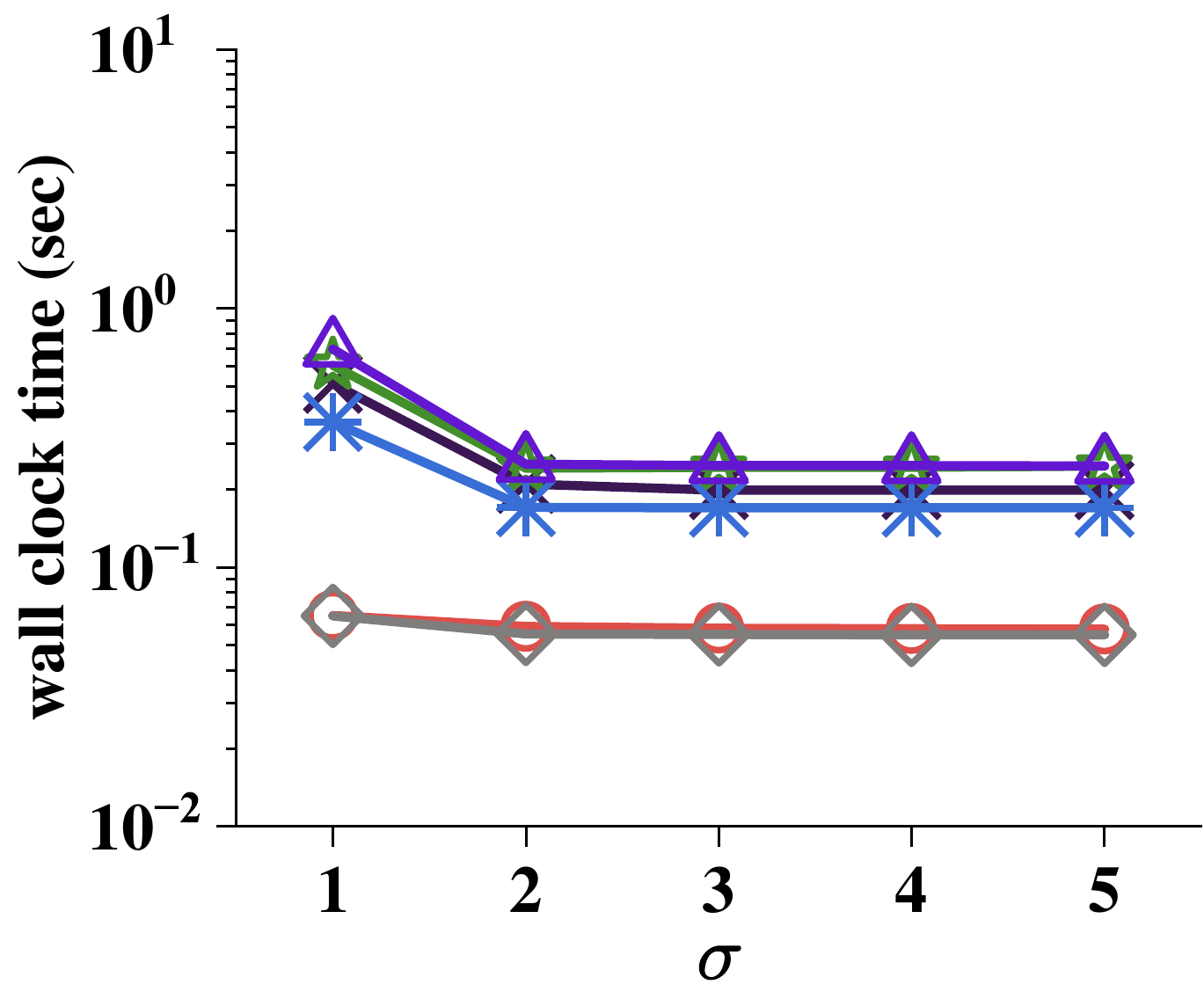}
        \label{subfig:effect_score_threshold}
    }\quad
    \hspace{-0.3cm}
    \subfigure[query keyword set size $|Q|$]{
        \includegraphics[height=2.7cm]{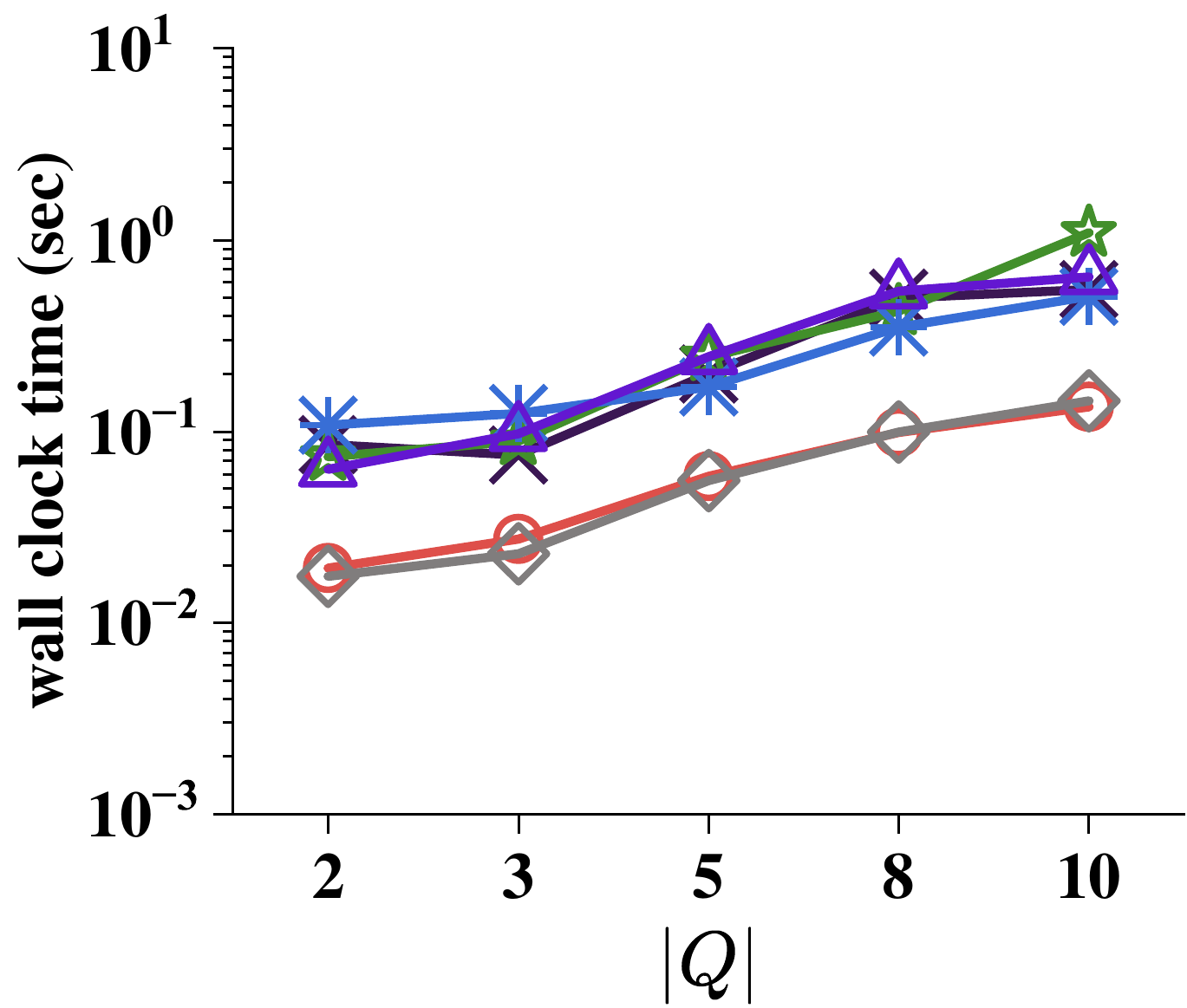}
        \label{subfig:effect_query_keyword_set_size}
    }\\
    \vspace{-0.3cm}
    \subfigure[keyword domain size $|\Sigma|$]{
        \includegraphics[height=2.7cm]{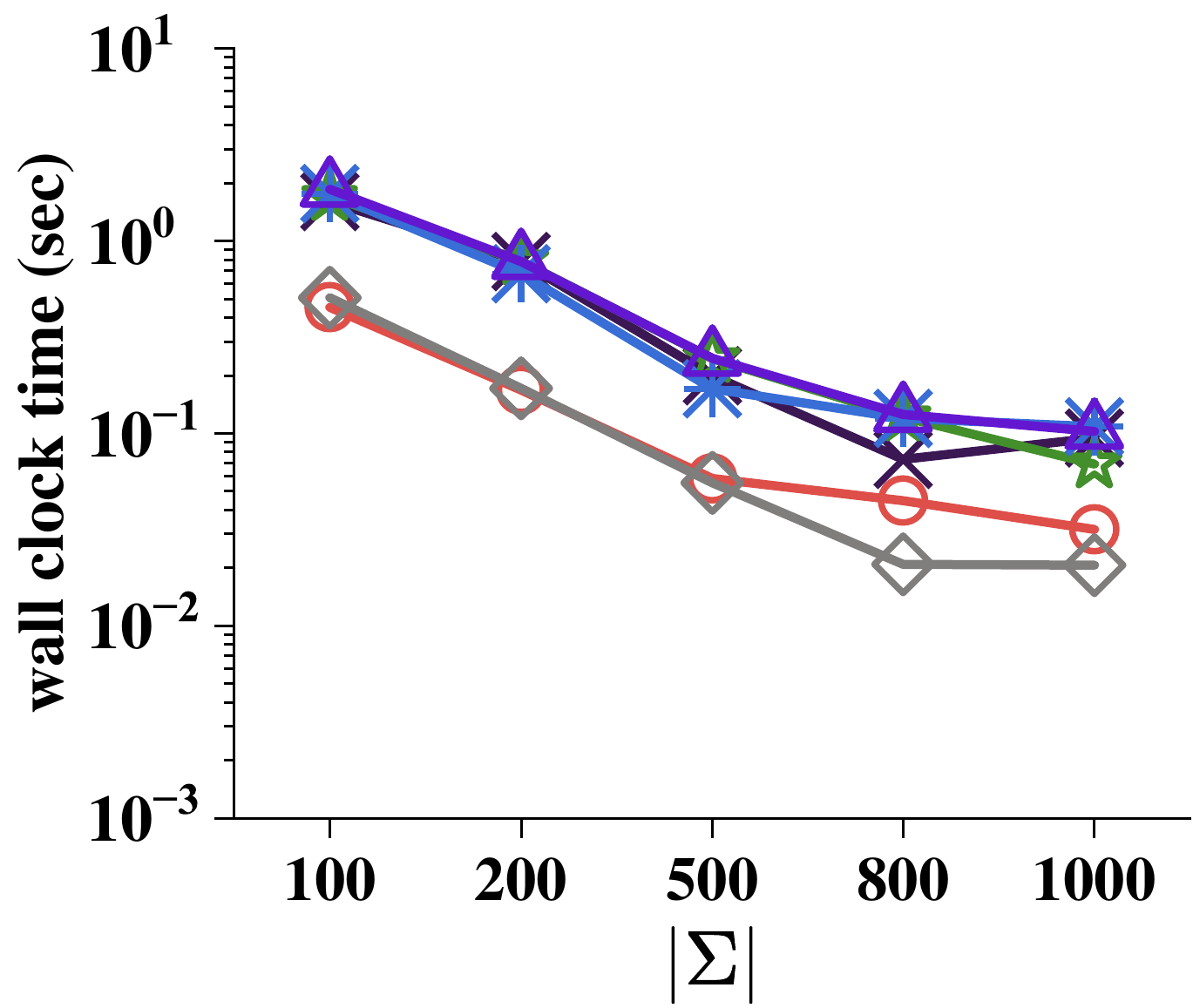}
        \label{subfig:effect_keyword_domain_size}
    }\quad
    \hspace{-0.3cm}
    \subfigure[\# of keywords/vertex $|v_i.K|$]{
        \includegraphics[height=2.7cm]{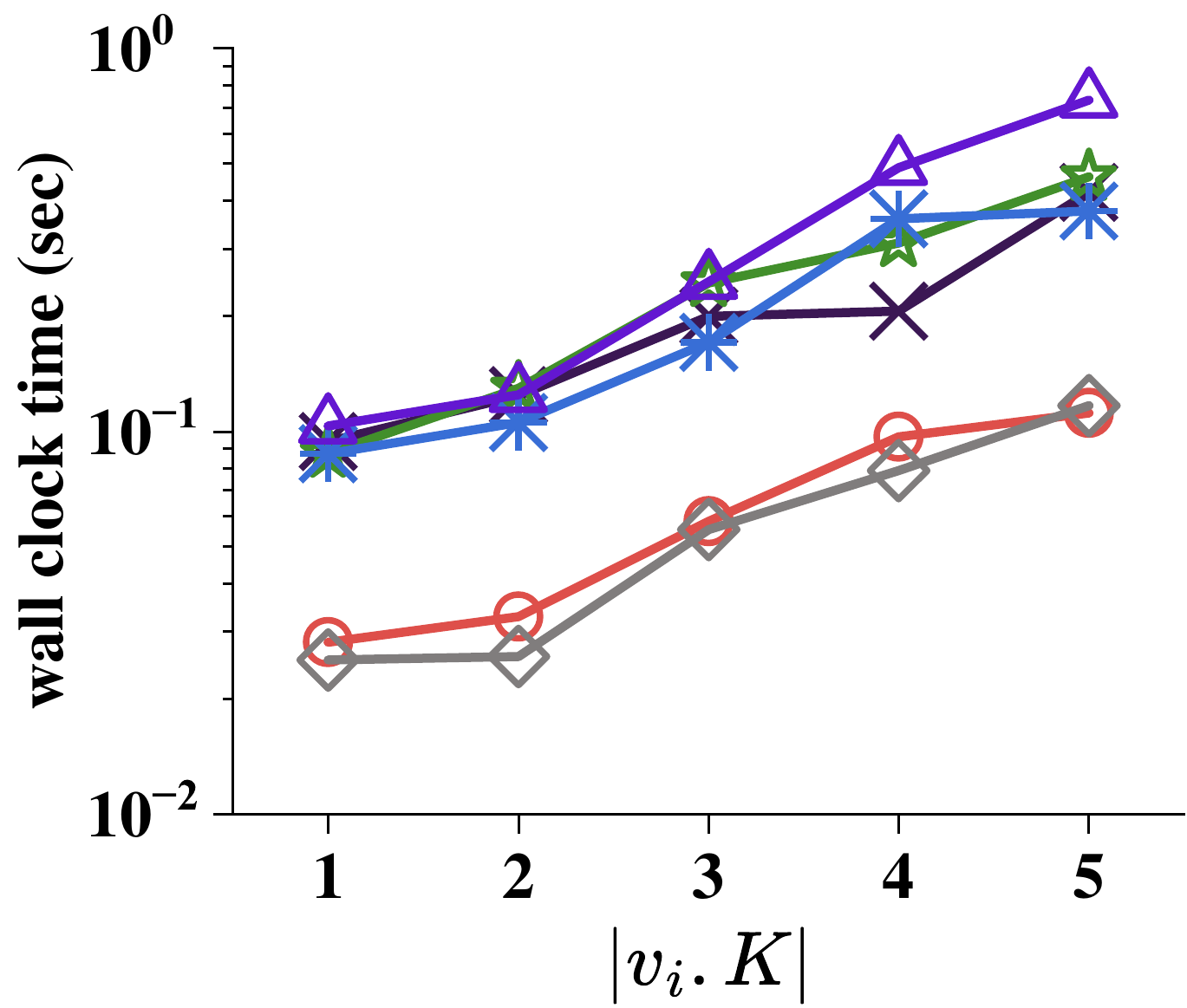}
        \label{subfig:effect_number_of_keywords_per_vertex}
    }\quad
    \hspace{-0.3cm}
    \subfigure[edge weight range]{
        \includegraphics[height=2.7cm]{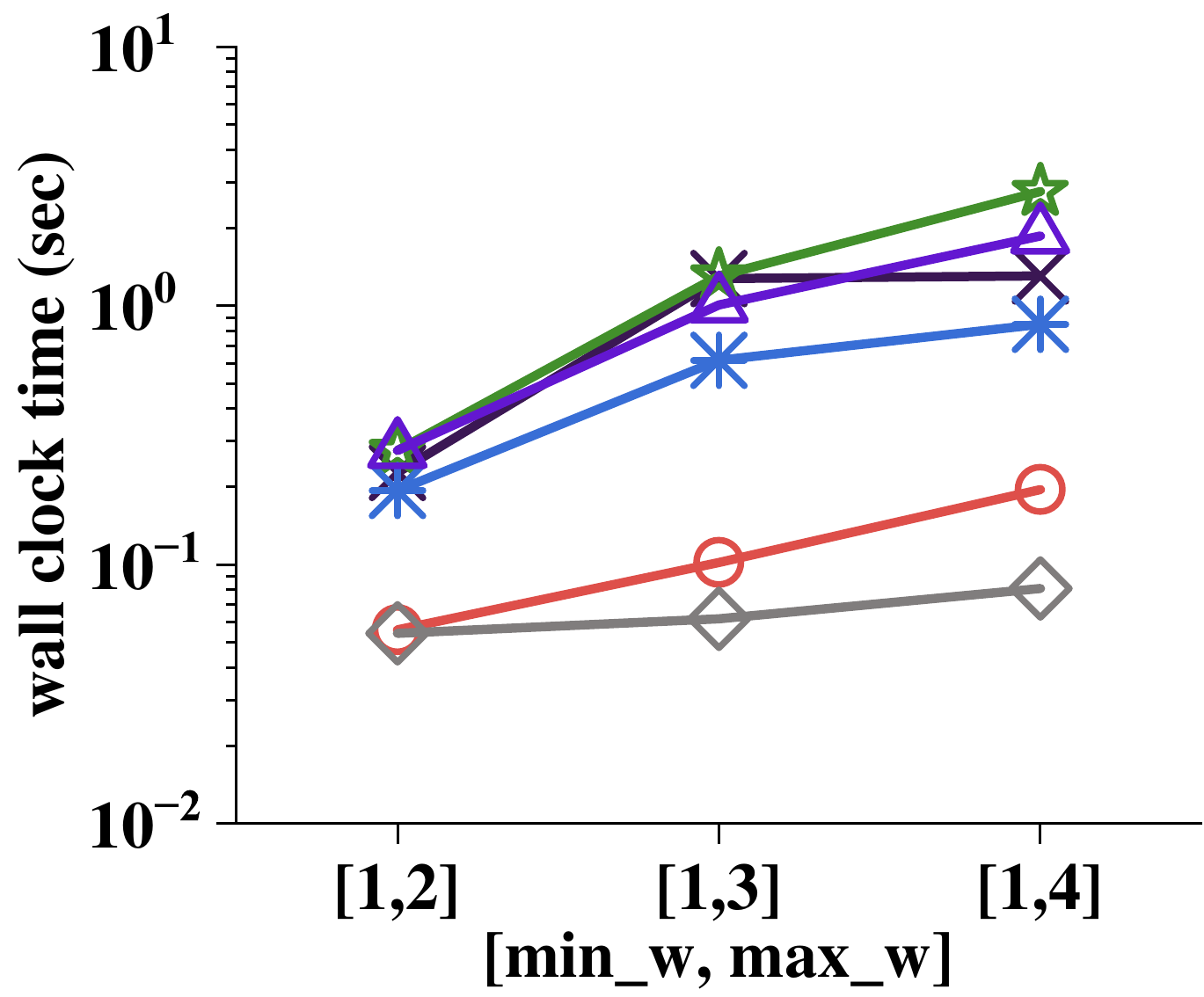}
        \label{subfig:effect_edge_weight_range}
    }\quad
    \hspace{-0.3cm}
    \subfigure[lower layer size $|L(G)|$]{
        \includegraphics[height=2.7cm]{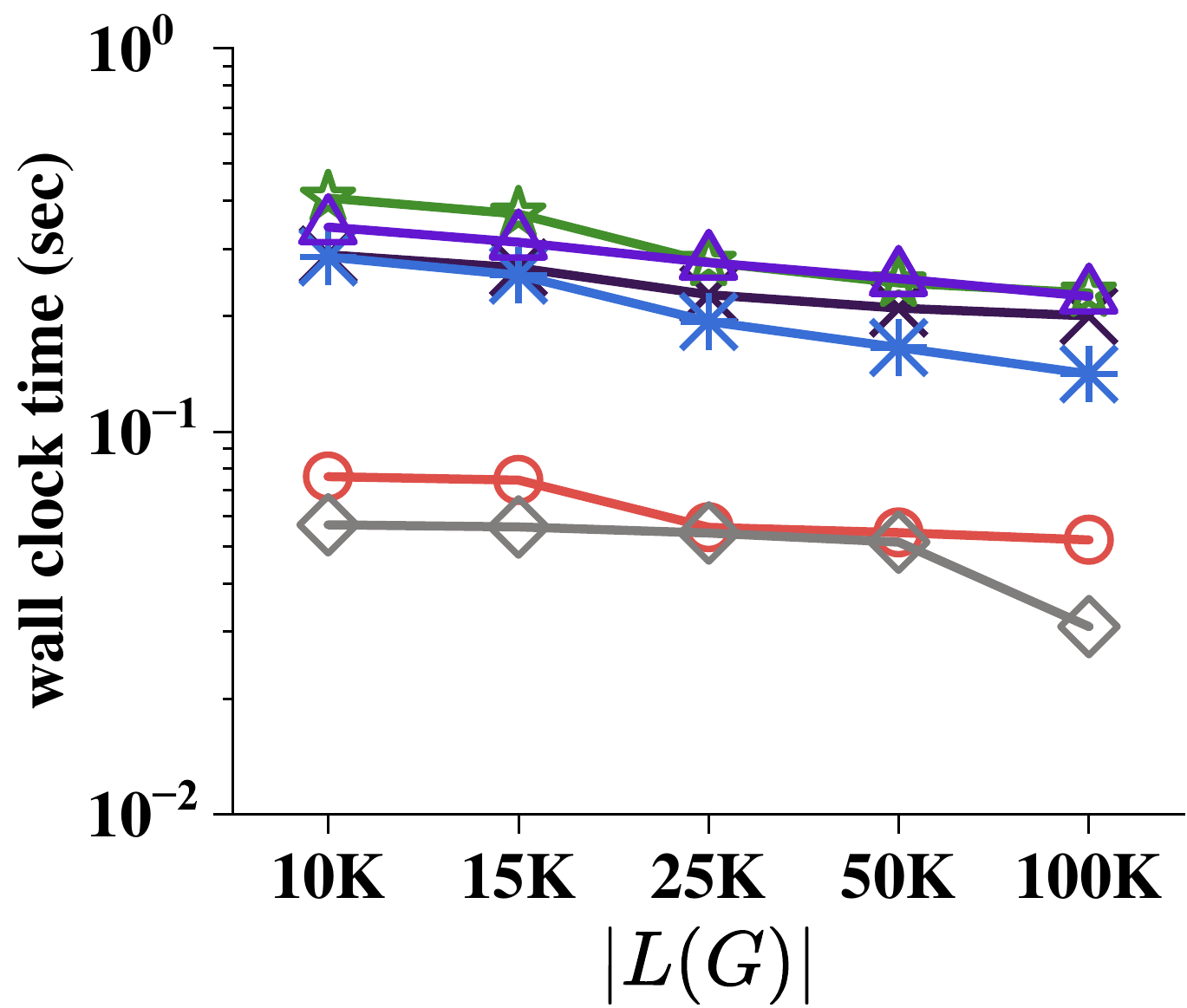}
        \label{subfig:effect_lower_layer_size}
    }\quad
    \hspace{-0.3cm}
    \subfigure[upper layer size $|U(G)|$]{
        \includegraphics[height=2.7cm]{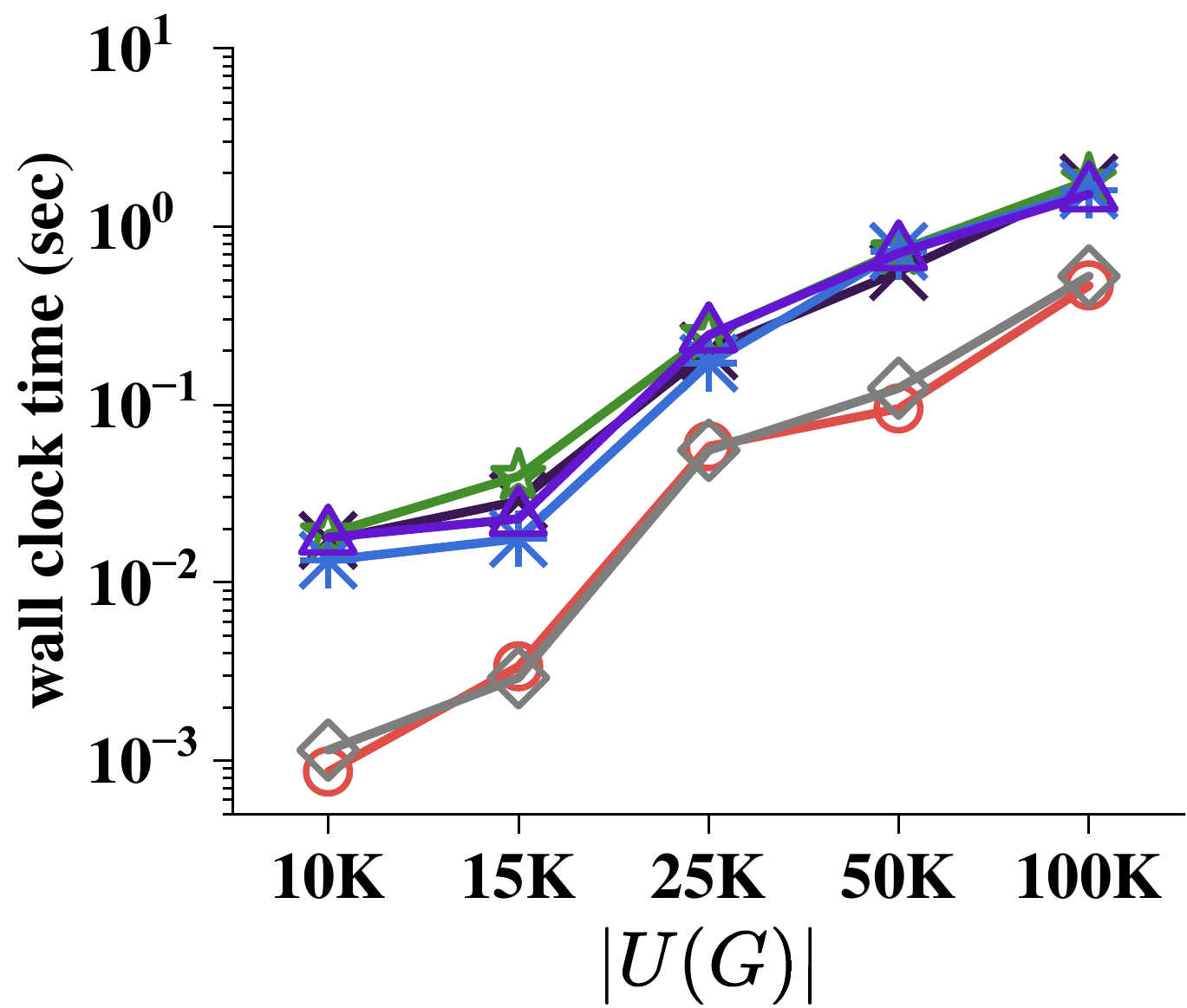}
        \label{subfig:effect_upper_layer_size}
    }\\
    \vspace{-3ex}
    \caption{\small The robustness evaluation of the CD-SBN approach.}
    \Description{\small The robustness evaluation of the CD-SBN approach.}
    \vspace{-2ex}
    \label{fig:efficiency}
\end{figure*}

\nop{

\begin{figure}[t!]
    \centering
        \includegraphics[height=6.0cm]{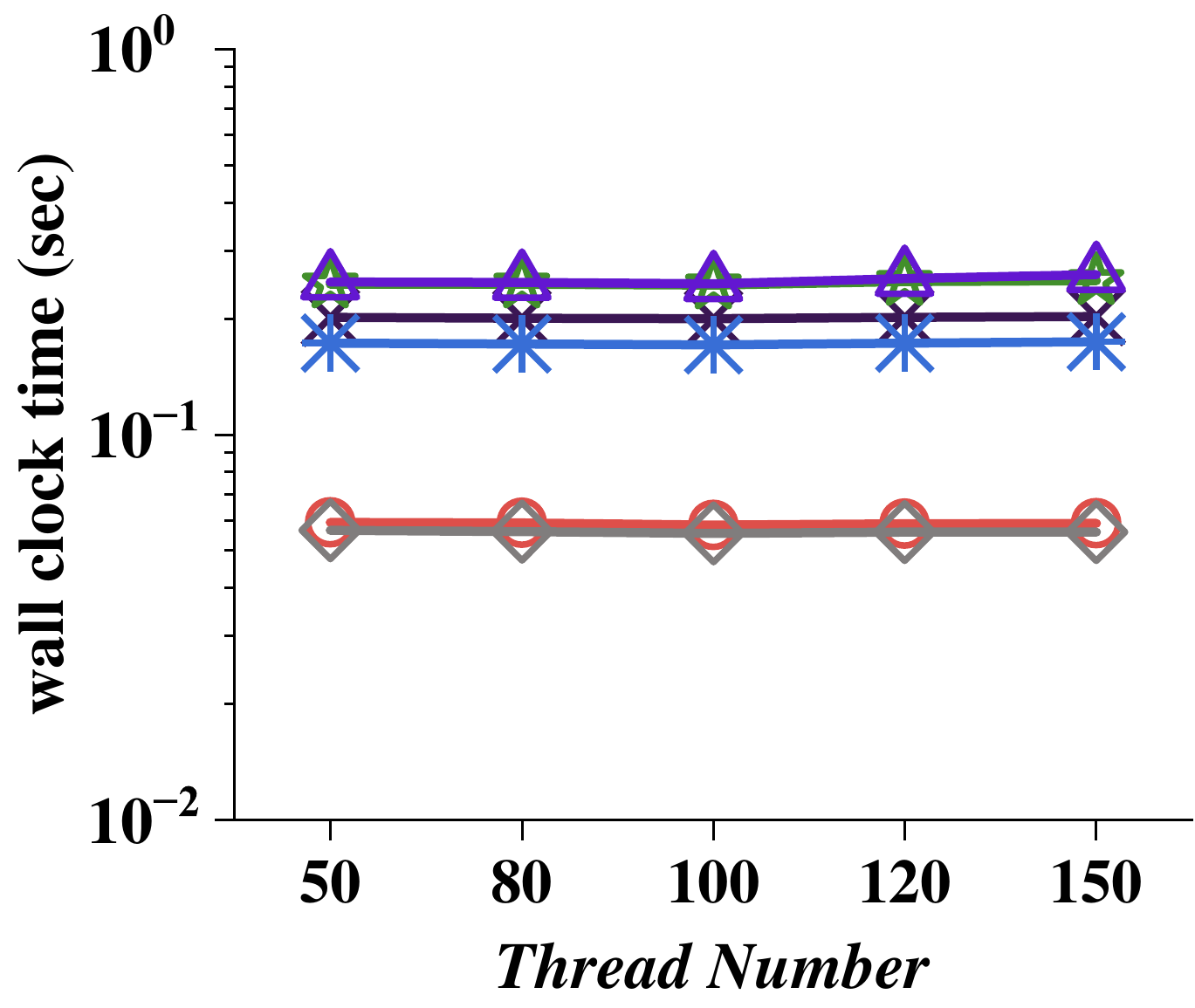}
    \vspace{-1ex}
    \caption{\small Time cost of continuous CD-SBN query processing with varying thread number.}
    \Description{\small Time cost of continuous CD-SBN query processing with varying thread number.}
    \label{fig:thread_num}
\end{figure}

\noindent {\bf Parameter Settings:}
Table~\ref{tab:parameters} depicts the parameter settings of our subsequent experiments, where default values are in bold. Each time we vary one parameter while keeping the others at their default values. {\color{blue}Moreover, since our continuous CD-SBN processing algorithm has no data dependencies during candidate detection, we use multithreaded optimization to speed up computation, with the number of threads set to $100$. As shown in~\ref{fig:thread_num},  All experiments are run on a machine with an AMD Ryzen Threadripper 3990X 64-Core processor, Ubuntu 22.04 OS, 192GB of memory, and a 600GB disk. All algorithms were implemented in C++ (C++17 standard) and compiled with g++ 11.4.0.
}
}


\noindent {\bf Parameter Settings:}
Table~\ref{tab:parameters} depicts the parameter settings of our subsequent experiments, where default values are in bold. Each time we vary one parameter while keeping the others at their default values. Moreover, we use multithreaded optimization to speed up the computation of candidates, where the number of threads is set to $100$ by default. All experiments are conducted on a machine equipped with an AMD Ryzen Threadripper 3990X 64-Core processor, running Ubuntu 22.04 OS, 192GB of memory, and a 600GB disk. All algorithms were implemented in C++ (C++17 standard) and compiled with g++ 11.4.0.

\noindent {\bf Research Questions:}
We design experiments for our CD-SBN algorithms to answer the following research questions (RQs):

\underline{\it RQ1 (Efficiency):} Can our proposed algorithms efficiently process the continuous CD-SBN problem?

\underline{\it RQ2 (Effectiveness):} Can our proposed pruning strategies effectively filter out false alarms of candidate communities during the CD-SBN processing?

\underline{\it RQ3 (Utility):} Do the resulting CD-SBN communities have practical significance for real-world applications?

\subsection{CD-SBN Performance Evaluation}
\noindent {\bf Continuous CD-SBN Efficiency (RQ1):}
Figure~\ref{fig:performance_total} compares the \textit{wall clock time} of our continuous CD-SBN processing algorithm with that of \textit{BBD} over real-world and synthetic bipartite graphs, where all parameters are set to default values in Table~\ref{tab:parameters}. The experiment results indicate that our continuous CD-SBN algorithm outperforms \textit{BBD} by more than two orders of magnitude, confirming the efficiency of our CD-SBN approach on real/synthetic graphs.

We also test the robustness of our CD-SBN approach with different parameters (e.g., $s$, $k$, $r$, $\sigma$, and $|Q|$) on synthetic graphs with different structural parameters (e.g., $|\Sigma|$, $|v_i.K|$, $[min\_w, max\_w]$, $|L(G)|$, and $|U(G)|$).

\underline{\it Effect of Sliding Window Size $s$:} Figure~\ref{subfig:effect_sliding_window_size} presents the performance of our continuous CD-SBN approach, by varying the size $s$ of the sliding window $W_t$ from $50$ to $200$, where other parameters are by default. From the experimental results, as $s$ increases, more edges are involved or have higher edge weights, resulting in more potential candidate communities. Therefore, a larger window size $s$ incurs a higher time cost, but the time cost of our CD-SBN approach remains low (i.e.,0.047  $\sim$ 0.726 $sec$).

\underline{\it Effect of Bitruss Support Threshold $k$:}
Figure~\ref{subfig:effect_truss_support_parameter} illustrates the continuous CD-SBN performance, where the bitruss support parameter $k$ varies from $3$ to $5$, and other parameters are set to default values. Since a larger support threshold $k$ yields greater pruning power, our CD-SBN approach can identify fewer candidate communities, thereby incurring lower time costs. The query cost remains low for different $k$ values (i.e., 0.055 $\sim$ 0.249 $sec$).


\begin{figure}[t!]\vspace{-2ex}
\subfigcapskip=-0.2cm
    \centering
    \subfigure[real-world graphs]{
        \includegraphics[height=2.5cm]{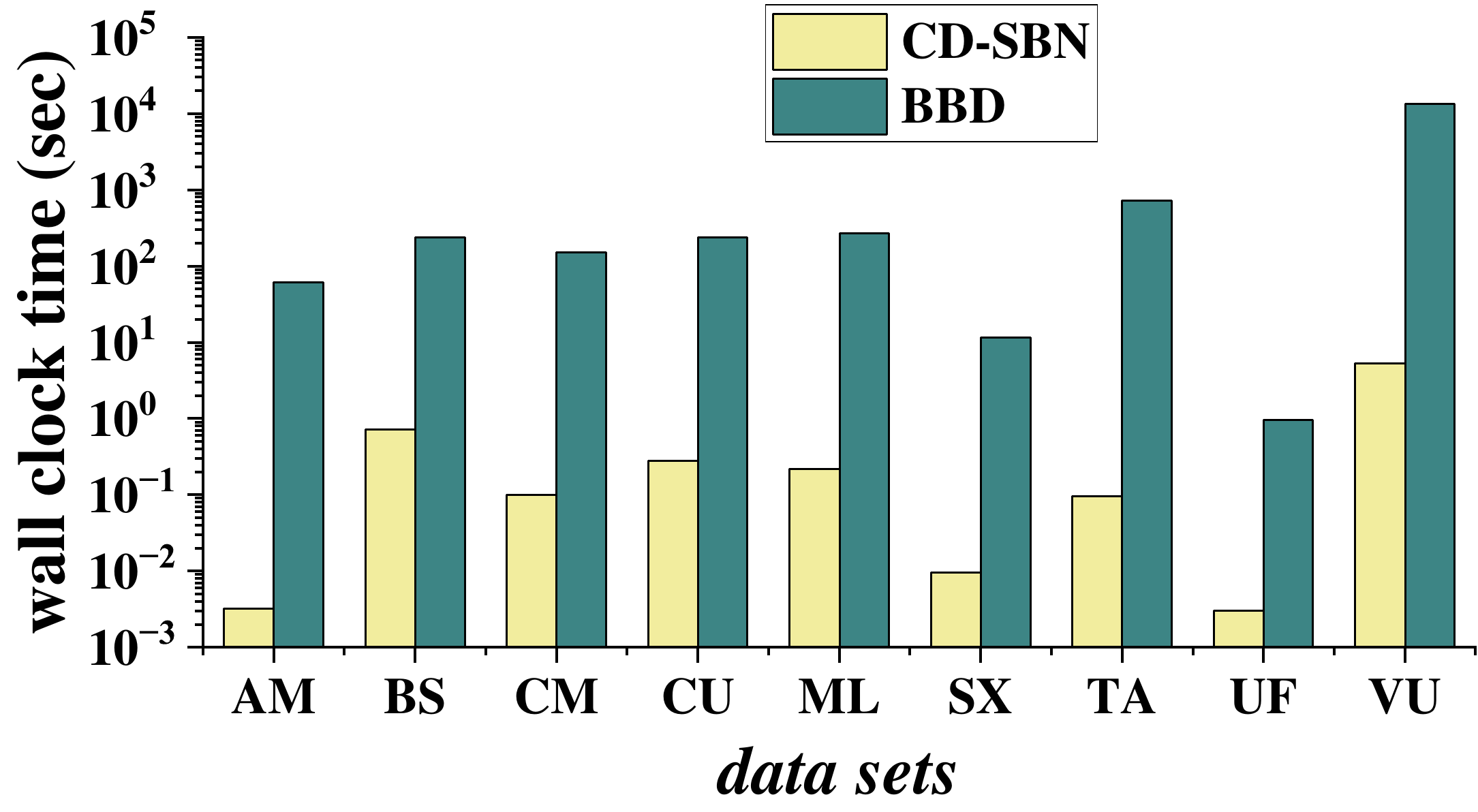}
        \label{subfig:performance_realworld}
    }\quad
    \hspace{-0.3cm}
    \subfigure[synthetic graphs]{
        \includegraphics[height=2.5cm]{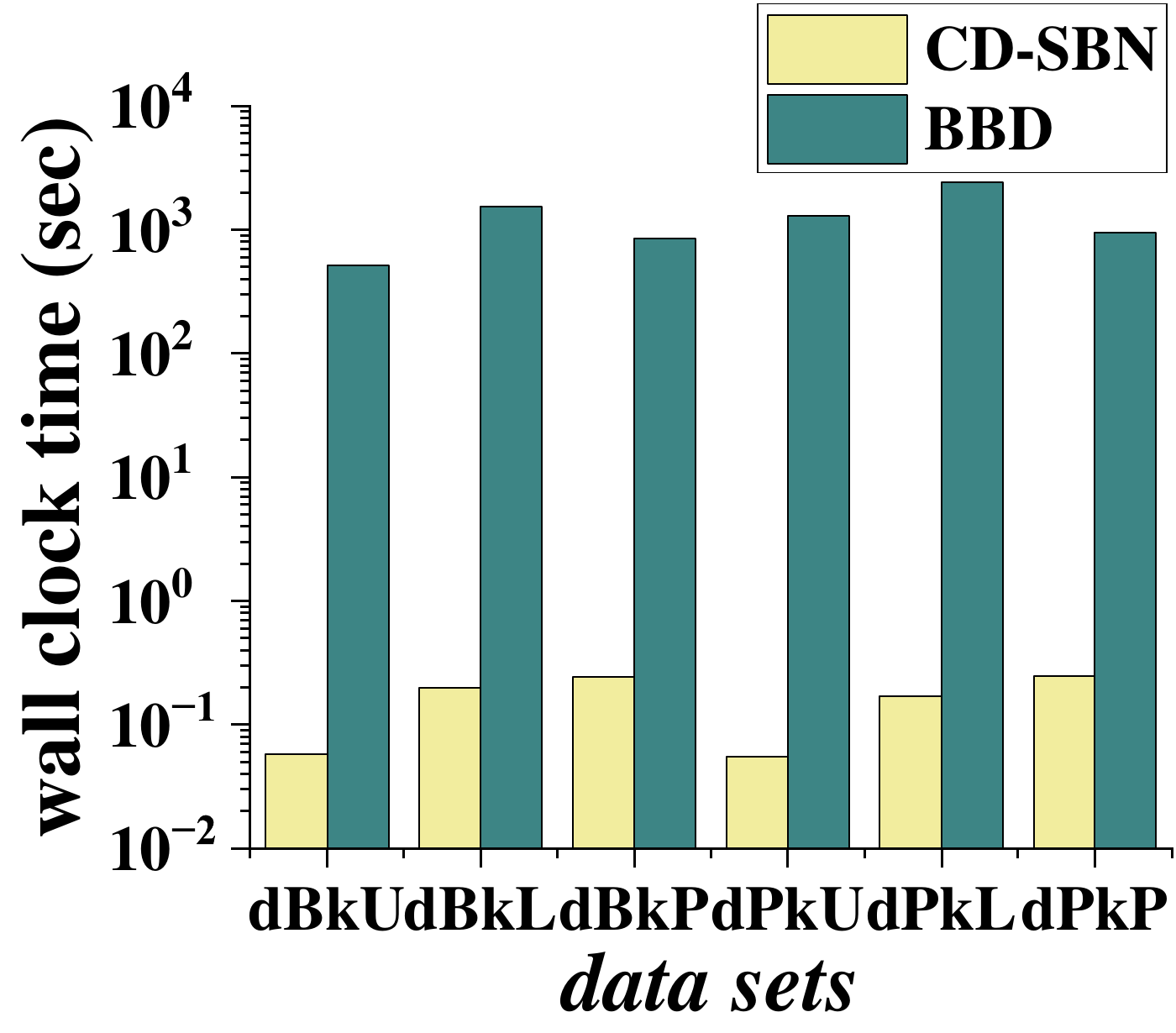}
        \label{subfig:performance_synthetic}
    }\quad
    \vspace{-4ex}
    \caption{\small The CD-SBN performance on real/synthetic graphs.}
    \Description{\small The CD-SBN performance on real/synthetic graphs.}
    \label{fig:performance_total}
\end{figure}

\underline{\it Effect of Radius $r$:}
Figure~\ref{subfig:effect_radius} shows the experimental results of our CD-SBN approach, by varying radius, $r$, the query community from $1$ to $3$, where default values are used for other parameters.
When the radius $r$ becomes larger, the scale of $2r$-hop subgraphs will increase, which results in higher time costs to obtain maximal $(k,r,\sigma)$-bitruss. 
Nevertheless, for different $r$ values, the \textit{wall clock time} remains small (i.e., around 0.032 $\sim$ 0.256 $sec$).

\nop{
\underline{\it Effect of Radius $r$:}
Figure~\ref{subfig:effect_radius} shows the experimental results of our continuous query processing approach for different radii $r$ ($r=1$, $2$, and $3$) of CD-SBN communities, while other parameters are all by default. When the query parameter radius $r$ becomes larger, the scale of $2r$-hop subgraphs will enlarge, which incurs higher time costs to obtain the maximal $(k,r,\sigma)$-bitruss. Regardless, the \textit{wall clock time} remains small (i.e., around 0.033 $\sim$ 0.375 $sec$) for different $r$ values.
}

\underline{\it Effect of Relationship Score Threshold $\sigma$:}
Figure~\ref{subfig:effect_score_threshold} varies the relationship score threshold, $\sigma$, from $1$ to $5$, where other parameters are set by default values.
Since higher $\sigma$ leads to higher pruning power and cohesiveness of candidate communities, the number of candidate communities decreases for larger $\sigma$. Thus, the CD-SBN time cost decreases for larger $\sigma$ values.
Overall, for different $\sigma$ values, the time cost remains low (i.e., 0.059 $\sim$ 0.697 $sec$).

\nop{
\underline{\it Effect of Relationship Score Threshold $\sigma$:}
Figure~\ref{subfig:effect_score_threshold} varies the relationship score threshold, $\sigma$, from $1$ to $5$, where all other parameters are set as default values. The higher sigma's value, the greater the candidate community's cohesion and pruning power. Hence, for the six synthetic graphs, as $\sigma$ increases, the \textit{wall clock time} decreases and remains low (i.e., 0.051$\sim$1.989 $sec$).
}

\underline{\it Effect of Query Keyword Set Size $|Q|$:}
Figure~\ref{subfig:effect_query_keyword_set_size} evaluates the CD-SBN performance, with different sizes, $|Q|$, of the query keyword set $Q$ from $2$ to $10$, where we use default values for other parameters. Intuitively, a larger set of query keywords $Q$ allows more item vertices to be included, resulting in lower pruning power by the keyword pruning. Consequently, this leads to the retrieval of more candidate communities for further filtering and refinement. Thus, from the figure, as $|Q|$ increases, the query cost also increases. Nonetheless, the wall clock time remains low (i.e., 0.007 $\sim$ 0.893 $sec$) for different $|Q|$ values. 

\nop{
\underline{\it Effect of Query Keyword Set Size $|Q|$:}
Figure~\ref{subfig:effect_query_keyword_set_size} evaluates the performance of our proposed CD-SBN continuous query processing algorithm when the query keyword set size, $|Q|$, is equal to $2$, $3$, $5$, $8$, and $10$, and the other parameters are set as default settings. 
A larger query keyword set leads to a larger candidate community set to filter out and refine, which incurs a higher time cost to maintain the continuous query answer. However, the \textit{wall clock time} to process the queries with different query keyword sets remains low (i.e., 0.019 $\sim$ 1.090 $sec$). 
}

\underline{\it Effect of Keyword Domain Size $|\Sigma|$:}
Figure~\ref{subfig:effect_keyword_domain_size} shows the wall clock time of our CD-SBN approach, where the keyword domain size, $|\Sigma|$, varies from $100$ to $1,000$ and other parameters are set to their default values.
In the figure, we find that as $|\Sigma|$ increases, our CD-SBN approach requires less time to retrieve CD-SBN communities, indicating the effectiveness of our keyword pruning for large keyword domains.
Overall, with different $|\Sigma|$ values, the time cost remains low (i.e., 1.761 $\sim$ 0.021 $sec$).

\nop{
\underline{\it Effect of Keyword Domain Size $|\Sigma|$:}
Figure~\ref{subfig:effect_keyword_domain_size} reflects the performance of the proposed CD-SBN method in synthetic graphs, where the keyword domain size, $|\Sigma|$, varies from $100$ to $1000$, and other parameters are by default. As the domain size of keywords grows, the pruning power of keyword pruning decreases, leading to a larger candidate community. Note that, although our approach needs to take more time to compute the CD-SBN community accurately, the \textit{wall clock time} of our approach remains low (i.e., 1.974 $\sim$ 0.009 $sec$) for all tested keyword domain sizes.
}

\underline{\it Effect of \# of Keywords Per Item Vertex $|v_i.K|$:}
Figure~\ref{subfig:effect_number_of_keywords_per_vertex} illustrates the CD-SBN performance for different numbers, $|v_i.K|$, of keywords per item vertex from $1$ to $5$, where we set other parameters to default values. Intuitively, more keywords in item vertices $v_i$ will increase the chance to match with query keyword set $Q$, which results in more/larger candidate communities to filter and refine. Therefore, with the increase of $|v_i.K|$, the time cost of our CD-SBN approach also increases, but, nevertheless, remains low (i.e., 0.028 $\sim$0.773 $sec$).

\nop{
\underline{\it Effect of \# of Keywords Per Item Vertex $|v_i.K|$:}
Figure~\ref{subfig:effect_number_of_keywords_per_vertex} records the performance of the CD-SBN continuous query processing approach on graphs with different numbers of keywords per item vertex, $|v_i.K|$, while other parameters are set by default. The more keywords the item vertices have, the easier it is for the keyword relevance of CD-SBN communities to be satisfied. Therefore, the candidate community is larger when processing insertion edges, resulting in longer computation time. 
Despite this, the \textit{wall clock time} of our approach remains low (i.e., 0.014 $\sim$ 0.821 $sec$).
}

\underline{\it Effect of Edge Weight Range $\left[min\_w, max\_w\right]$:}
Figure~\ref{subfig:effect_edge_weight_range} examines the CD-SBN performance with different edge weight ranges, $\left[min\_w, max\_w\right]$ (i.e., $[1,2]$, $[1,3]$, and $[1,4]$), where default values are used for other parameters.
When the range of edge weights expands, more/larger candidate communities will be retrieved and refined. Thus, the query cost increases as the edge weight interval $\left[min\_w, max\_w\right]$ becomes wider. Nonetheless, the wall clock time remains low (i.e., 0.0543$\sim$2.761 $sec$).

\nop{
\underline{\it Effect of Edge Weight Range $\left[min\_w, max\_w\right]$:}
Figure~\ref{subfig:effect_edge_weight_range} illustrates the performance of our continuous CD-SBN query processing algorithm on different graphs whose edge weights are generated in different ranges following \textit{Gaussian} distribution, and other parameters are set by default values. We tested the edge weights in $\left[1,2\right]$, $\left[1,3\right]$ and $\left[1,4\right]$. As the range of edge weights expands (i.e., the maximum value increases), the scale of communities grows, resulting in higher time costs. However, the \textit{wall clock time} remains low (i.e., 0.0543 $\sim$ 2.761 $sec$ ).
}

\underline{\it Effect of Lower Layer Size $|L(G)|$:}
Figure~\ref{subfig:effect_lower_layer_size} tests the scalability of our proposed CD-SBN approach by varying the lower layer size, $|L(G)|$, from $10K$ to $100K$, where default values are set for other parameters.
From the figure, we can see that the query cost decreases as $|L(G)|$ increases. This is because, with a constant average degree of each user vertex, a larger $|L(G)|$ leads to a more sparse bipartite network, which incurs fewer candidate communities and lower query costs. 
The time cost remains low (i.e., 0.057$\sim$ 0.230 $sec$).

\nop{
\underline{\it Effect of Lower Layer Size $|L(G)|$:}
Figure~\ref{subfig:effect_lower_layer_size} reveals the scalability of our proposed CD-SBN approach when the lower layer size varies from $10K$ to $100K$, and other parameters are set with default values. Since we keep the average degree of each user vertex constant across bipartite networks with different lower-layer sizes, the bipartite network becomes sparser as the lower-layer size increases. Thus, the \textit{wall clock time} decreases (i.e., 0.057 $\sim$ 0.230 $sec$) with the increasing lower layer size and remains low.
}


\begin{figure}[t!]
    \centering
        \includegraphics[height=3.0cm]{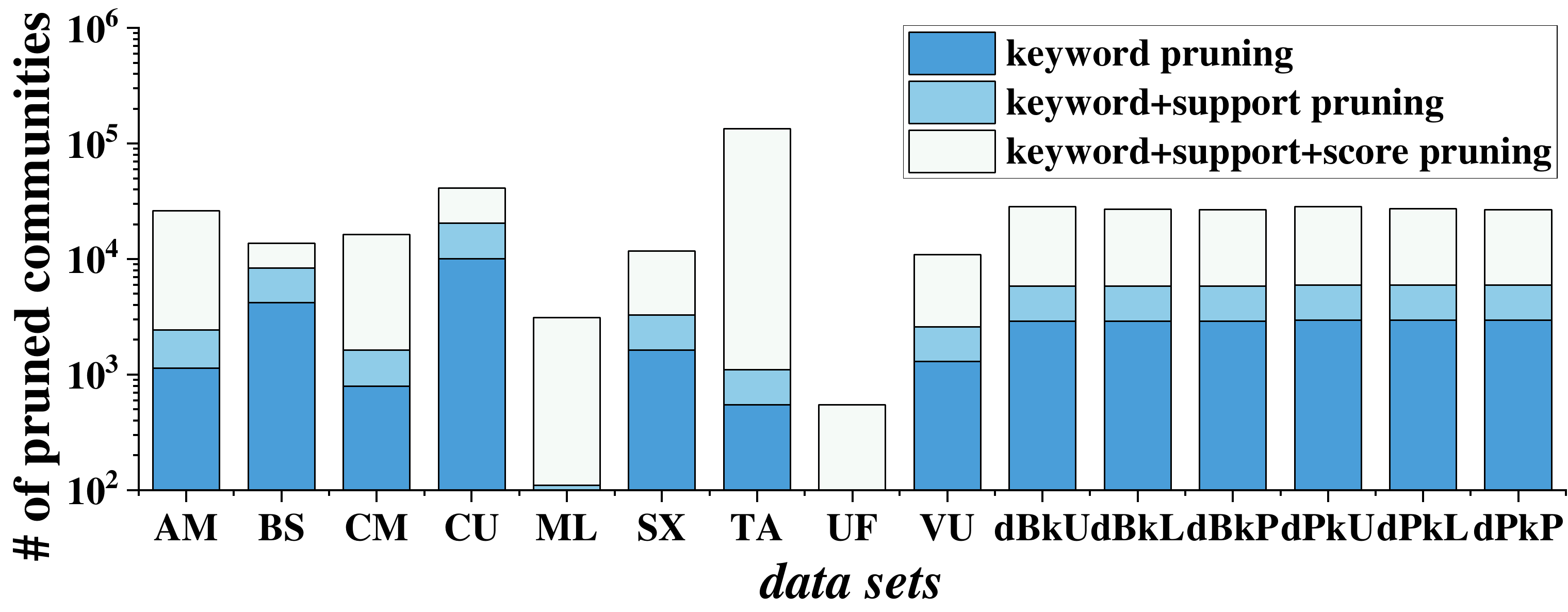}
    \vspace{-2ex}
    \caption{\small The ablation study of the pruning power for the snapshot CD-SBN problem.}
    \Description{\small Ablation study of the pruning power for the snapshot CD-SBN problem.}
    \label{fig:ablation_study}
\end{figure}

\underline{\it Effect of Upper Layer Size $|U(G)|$:}
Figure~\ref{subfig:effect_upper_layer_size} evaluates the scalability of our CD-SBN approach with different upper layer sizes, $|U(G)|$, from $10K$ to $100K$, where other parameters are set by default.
In the figure, as $|U(G)|$ increases, the CD-SBN time cost increases due to more retrieved candidate communities. Nonetheless, the time cost remains low (i.e., 0.001 $\sim$ 1.815 $sec$), which confirms the good scalability of our CD-SBN approach for large $|U(G)|$.

\nop{
\underline{\it Effect of Upper Layer Size $|U(G)|$:}
Figure~\ref{subfig:effect_upper_layer_size} tests the scalability of our CD-SBN method with varying upper layer sizes, $|U(G)|$, from $10K$ to $100K$, where other parameters are assigned with default values. In this figure, the \textit{wall clock time} increases smoothly as the size of the upper layer increases (i.e., 0.001 $\sim$ 1.815 $sec$), which confirms the scalability of our CD-SBN method for large upper layer sizes.
}

\noindent {\bf Ablation Study on the Snapshot CD-SBN Pruning Strategies (RQ2):}
To test the effectiveness of our proposed pruning methods, we conduct an ablation study of our snapshot CD-SBN pruning strategies on real-world/synthetic graphs, with all parameters set to default values. We execute the snapshot CD-SBN processing algorithm by adding one pruning strategy at a time, which forms 3 experimental groups: (1) \textit{support pruning} only, (2) \textit{support + score pruning}, and (3) \textit{support + score + keyword pruning}.
Figure~\ref{fig:ablation_study} reports the number of pruned candidate communities by using different pruning method groups.
In this figure, we observe that the number of pruned communities increases by approximately an order of magnitude as we apply more pruning strategies, indicating the effectiveness of our proposed strategies.

\noindent {\bf CD-SBN Community Case Study (RQ3):}
To evaluate the utility of our CD-SBN results, we carry out a case study to compare our proposed $(k,r,\sigma)$-bitruss pattern with the biclique \cite{lin2018TheParameterized} (which is a variant of bipartite community structure) over \textit{BibSonomy} (BS) graph data. The entire graph shown in Figure~\ref{fig:case_study} is a $(4,2,3)$-bitruss containing 3 users (in blue) and 9 items (in blue or red), whereas the 3 users, 4 items (in blue) form a $3 \times 4$ biclique. We can see that the $3 \times 4$ biclique shares the same 3 users as the $(4,2,3)$-bitruss, but the $3 \times 4$ biclique omits some items that are connected to users with large weights (e.g., 6). In contrast, our $(k,r,\sigma)$-bitruss captures these items, and can recommend them to users who are not connected to them (e.g., recommending items $v\_449272$, $v\_449273$, $v\_449277$, $v\_449275$, and $449276$ to user $u\_4954$). This confirms the utility of our proposed CD-SBN community semantics, which can recommend items to users who are potentially more interested.

\section{Related Work}
\label{sec:related_work}

\noindent{\bf Community Search/Detection Over Unipartite Graphs:} As fundamental operations in social network analysis, the \textit{community search} (CS) and \textit{community detection} (CD) have been extensively studied for the past few decades. 
In unipartite graphs, different community semantics have been proposed, including minimum degree~\cite{cui2014LocalSearchCommunitiesa}, $k$-core~\cite{sozio2010CommunitysearchProblemHowa}, $k$-clique~\cite{cui2013OnlineSearchOverlapping}, and $(k,d)$-truss~\cite{huang2017attribute,al2020topic}.
\nop{
Previous works on community detection usually retrieved all communities based on link information \cite{newman2004finding,fortunato2010community} or clustering techniques \cite{xu2012model,conte2018d2k,veldt2018correlation}.
}
In contrast, our work is conducted on bipartite graphs with two distinct types of nodes, and our proposed $(k,r,\sigma)$-bitruss considers the high connectivity of community structure, edge weights, and query keywords, which is more challenging.

\noindent{\bf Dynamic Community Search/Detection Over Unipartite Graphs:} 
Due to the importance of hidden trends in streaming data, CS/CD over large-scale dynamic graphs has been widely studied.
\nop{Some recent studies aim to find communities over dynamic social networks, e.g., \cite{xu2022efficient} searches for a \textit{Triangle-connected $k$-Truss Community} where edges are contained in adjacent triangles and \cite{tang2024ReliabilityDrivenLocalCommunity, tang2022ReliableCommunitySearch} work on the $k$-core with edge weights higher than $\theta$ (i.e., $(\theta, k)$-core).}
Some recent studies aim to identify communities in dynamic social networks, e.g., the \textit{Triangle-connected $k$-Truss Community}~\cite{xu2022efficient} and the $(\theta, k)$-core~\cite{tang2024ReliabilityDrivenLocalCommunity, tang2022ReliableCommunitySearch}.
However, our CD-SBN problem aims to detect communities over streaming bipartite graphs, which is a more complicated task.

\begin{figure}[t]
    \centering
    \includegraphics[height=4.5cm]{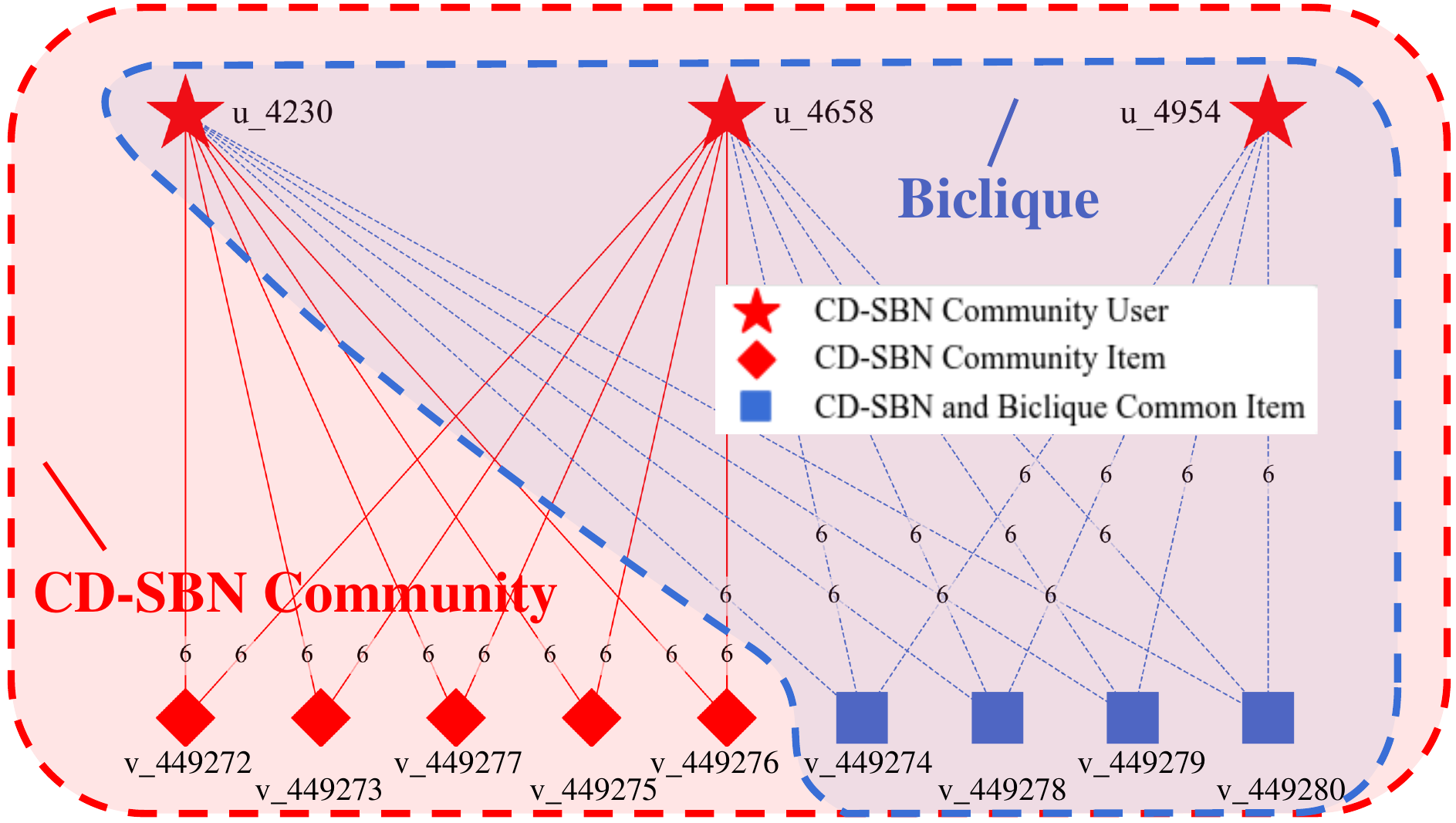}
    \vspace{-1ex}
    \caption{\small A case study of $(k,r,\sigma)$-bitruss vs. biclique \cite{lin2018TheParameterized} (query keyword set $Q$ = \{``Bookmarks''\}).}
    \Description{\small Case study of $(k,r,\sigma)$-bitruss vs. biclique \cite{lin2018TheParameterized} (query keyword set $Q$ = \{``Bookmarks''\}).}
    \label{fig:case_study}
\end{figure}

\nop{
\noindent{\bf Dynamic Community Search/Detection Over Unipartite Graphs:} 
In recent years, CS/CD over large-scale dynamic graphs has become an increasingly important research problem due to the presence of hidden trends and patterns in streaming data.
Some existing studies on dynamic CS/CD aim to find potential communities over dynamic social networks with different specified community structures to gain valuable insights, e.g., over unipartite dynamic graphs, \cite{xu2022efficient} searches for and maintains a \textit{Triangle-connected $k$-Truss Community} where edges should be contained in adjacent triangles and \cite{tang2024ReliabilityDrivenLocalCommunity, tang2022ReliableCommunitySearch} work on the $k$-core with edge weights higher than $\theta$ (i.e., $(\theta, k)$-core). Unlike these scenarios, our proposed CD-SBN problem focuses on detecting communities over streaming bipartite graphs, which requires more complex calculations and greater computational overhead.
}

\noindent{\bf Community Search/Detection Over Bipartite Graphs:}
Prior works on CD over \textit{static} bipartite graphs include model-based ~\cite{qing2023CommunityDetectionWeighted, rohe2016CoclusteringDirectedGraphs} and modularity-based methods~\cite{zhou2022EndtoendModularitybasedCommunity, bouguessa2020BiNeTClusBipartiteNetwork}, which require recalculations upon graph updates and are computationally expensive.
Moreover, several dense-subgraph-based methods did not consider edge weights, including the \textit{size-bounded $(\alpha, \beta)$-core}~\cite{zhang2024SizeboundedCommunitySearch, wang2022EfficientPersonalizedMaximum}, \textit{personalized $k$-wing}~\cite{abidi2023SearchingPersonalizedkWing}, and \textit{Pareto-optimal $(\alpha, \beta)$-community}~\cite{zhang2021ParetooptimalCommunitySearch}. Other works focused on weighted bipartite graphs and considered \textit{significant $(\alpha, \beta)$-community} (SC)~\cite{wang2021EfficientEffectiveCommunity,wang2023DiscoveringSignificantCommunities} and  \textit{$(\alpha, \beta)$-attributed weighted community} ($(\alpha, \beta)$-AWC)~\cite{li2022AlphaBetaAWCSAlpha}.
However, the SC and the $(\alpha, \beta)$-AWC consider different community semantics from ours, so they cannot be directly utilized to solve our CD-SBN problem.
Furthermore, Zhou et al.\cite{zhou2022EndtoendModularitybasedCommunity} employed a learning-based clustering approach to identify communities on bipartite graphs; however, the resulting communities are sometimes unstable in terms of structure and size.
The aforementioned works focused on \textit{static} bipartite graphs, but do not apply to our CD-SBN problem on streaming bipartite graphs.
Moreover, some prior works \cite{wang2023DiscoveringSignificantCommunities, abidi2023SearchingPersonalizedkWing} constructed a dynamically updated index to solve CS/CD on streaming bipartite graphs, requiring re-traversal of the entire index. 
In contrast, our CD-SBN approach incrementally maintains the result set upon edge updates with lower costs.

\section{Conclusion}
\label{sec:conclusion}
In this paper, we propose a novel CD-SBN problem to detect communities with user-specified query keywords and high structural cohesiveness in both snapshot and continuous scenarios over streaming bipartite graphs. To efficiently tackle the CD-SBN problem, we propose effective pruning strategies to eliminate false alarms in candidate communities and design a hierarchical synopsis to facilitate the processing of CD-SBN. We also develop efficient algorithms to enable snapshot/continuous CD-SBN answer retrieval by quickly accessing the synopsis and applying our proposed pruning methods. Comprehensive experiments confirmed the effectiveness and efficiency of our proposed CD-SBN approaches over real/synthetic streaming bipartite networks upon graph updates.

\clearpage
\bibliographystyle{ACM-Reference-Format}
\balance
\bibliography{cites}



\end{document}